\documentclass[
               format=sigconf,  %
               screen=true,
               acmthm=true,
               nonacm=true,
               ]{acmart}

\setcopyright{acmlicensed}
\copyrightyear{2024}

\usepackage{tikz,graphicx}
\usepackage{tikzmacros}

\usepackage{mathtools}

 \AtEndPreamble{%
\theoremstyle{acmplain}

}

\usepackage{subfigure}

\usepackage[ruled]{algorithm2e} %

\SetAlFnt{\small} \SetAlCapFnt{\small}
\SetAlCapNameFnt{\small}
\SetAlCapHSkip{0pt}
\IncMargin{-\parindent}

\usepackage[full]{complexity}

\usepackage[inline]{enumitem}

\usepackage{hyperref}
\usepackage[capitalize,noabbrev]{cleveref}

\usepackage{thmtools}
\usepackage{thm-restate}

\usepackage{macros}

\begin{document}

\title{Bounded-Memory Strategies in Partial-Information Games}

\author{Sougata Bose}
\orcid{0000-0003-3662-3915}
\email{sougata.bose@liverpool.ac.uk}
\affiliation{%
	\institution{University of Liverpool}
	\city{Liverpool}
	\country{UK}
}
\author{Rasmus Ibsen-Jensen}
\orcid{0000-0003-4783-0389}
\email{r.ibsen-jensen@liverpool.ac.uk}
\affiliation{%
	\institution{University of Liverpool}
	\city{Liverpool}
	\country{UK}
}
\author{Patrick Totzke}
\orcid{0001-5274-8190}
\email{totzke@liverpool.ac.uk}

\affiliation{%
    \institution{University of Liverpool}
    \city{Liverpool}
    \country{UK}
}

\renewcommand{\shortauthors}{S.~Bose, R.~Ibsen-Jensen, P.~Totzke}

\begin{abstract}
We study the computational complexity of solving stochastic games with mean-payoff objectives.
Instead of identifying special classes in which simple strategies are sufficient to play $\epsilon$-optimally, or form $\epsilon$-Nash equilibria,
we consider general partial-information multiplayer games and 
ask what can be achieved with (and against) finite-memory strategies up to a {given} bound on the memory. 

We show $\NP$-hardness for approximating zero-sum values,
already with respect to memoryless strategies and for 1-player reachability games.
On the other hand, we provide upper bounds
for solving games of any fixed number of players $k$.
We show that one can decide in polynomial space if, 
for a given $k$-player game, $\eps\ge 0$ and bound $b$,
there exists an $\eps$-Nash equilibrium in which all strategies use at most $b$ memory modes.

For given $\eps>0$, finding an $\epsilon$-Nash equilibrium
with respect to $b$-bounded strategies can be done in $\FNP^{\NP}$.
Similarly for 2-player zero-sum games, finding a $b$-bounded strategy that, against all $b$-bounded opponent strategies,
guarantees an outcome within $\epsilon$ of a given value, 
can be done in $\FNP^{\NP}$. 
Our constructions apply to parity objectives with minimal simplifications.

Our results improve the status quo in several well-known special cases of games.
In particular, for $2$-player zero-sum concurrent mean-payoff games,
one can approximate ordinary zero-sum values (without restricting admissible strategies) in
$\FNP^{\NP}$. 

\end{abstract}



\maketitle

\newcommand{\mypar}[1]{\subsubsection*{#1}}

\section{Introduction}
\label{sec:intro}
We study stochastic games of infinite duration, where $k$ players jointly move a pebble through a finite directed graph.
In every round, all players independently select, possibly at random, one out of finitely many actions.
The pebble is moved according to a fixed distribution for the current state and chosen action vector, the players receive an associated reward and the next round begins.
In general partial-information games, players do not observe the choices or rewards of other players nor the current state of the pebble. Instead, at the beginning of each round, they receive a signal that may or may not include that information.
This subsumes classical settings such as 
perfect-information games (signals include the current state, rewards, and actions of others), turn-based games (in every round only one player's action affects the outcome),
and non-stochastic games (all distributions are Dirac).
We focus on games with mean-payoff (aka {limit-average}, where 
players aim to maximise their expected long-term average rewards).
We discuss in \cref{ssec:imp-parity} how our constructions
can be applied to parity games as well with minor simplifying adjustments.

We want to solve such games, i.e., compute ($\eps$)-Nash equilibria,
zero-sum values and witnessing strategies.
However, partial-information %
games are very expressive
and not guaranteed to have (Nash, or $\eps$-Nash) equilibria or zero-sum values \cite{SB2016}. 
Moreover, already for rather restricted subclasses that do admit $\eps$-optimal strategies,
finding them is undecidable: this is the case
for single-player, reachability games with a unique signal, aka~partially observable Markov decision processes (POMDPs) \cite{MHC1999,CSZ2022}. 
Here,
finding the best finite-memory strategy (one that is representable by a finite automaton) without providing a bound on its size up front is complete for the recursive-enumerable problems
\cite{CSZ2022}.
For the case of partial-information stochastic parity games, where one player has perfect-information, finding the best finite-memory strategy is EXPTIME-complete \cite{CDNV2014}.
Even for perfect-information games such as the classic ``Big Match''~\cite{G1957,BF1968},
finite-memory strategies can be much worse than general ones \cite{HIN2018}. 
For this reason, and motivated by applications, much work has focused on identifying special cases of games in which simple strategies are sufficient to play $\eps$-optimally, or to guarantee the existence of $\epsilon$-Nash equilibria
\cite{G1957,EM1973,MN1981,FSTV1991,Put1994,MS1999,SS2002,GZ2005,HKMT2011,HPPR2018,BRORV2022,BRV2023}.

In this paper, we study the complexity of solving games \emph{with respect to bounded-memory strategies},
meaning that all players are required to play according to a finite-memory strategy
where the number of distinct memory modes is bounded a priori.
We ask what outcomes players can still guarantee, and how,
if one assumes that only strategies up to the given bound are admissible.
This approach lets us study ``realistic'' scenarios that require bounded representations for strategies.
Moreover, for many classes of games,
values/equilibria with respect to $b$-bounded strategies coincide with ordinary zero-sum values/equilibria.
For example, concurrent reachability games \cite{HKM2009,HKMT2011,FMB2013,HIM2014} admit $\eps$-optimal strategies that are memoryless ($b=1$). Consequently, our upper bounds directly apply to such special cases.

There remain significant challenges for solving games
even with known bounds on the number of states and memory modes in admissible strategies: (1)~optimal strategies/Nash-equilibria might not exist \cite{E1957}, (2)~values/outcomes of equilibria can be irrational~\cite{EY2008} and can therefore at best be approximated, and (3) strategies with rational distributions may require double-exponentially small probabilities to be ($\eps$-)optimal \cite{HKM2009}.
This last point is especially problematic since it means it requires more than polynomial time to guess/write down a good strategy with binary encoded values.

\mypar{Contributions}

We provide new complexity bounds for approximating $\eps$-optimal strategies and $\eps$-equilibria in partial-information games with bounded strategies.
We give an overview here; precise claims are delayed to \cref{sec:results} after introducing notations.

First, we give $\NP$ and $\coNP$ lower bounds for approximating the value achievable by bounded strategies under very weak assumptions.
More precisely, we show that it is \NP-hard to check if in a given 1-player reachability game the value achieved by memoryless strategies exceeds a given threshold.
This holds even if the threshold is unary encoded, and even though the games we construct admit optimal memoryless strategies that do not require randomisation.
As a simple corollary, we obtain \coNP-hardness for $2$-player, zero-sum games.
An $\NP$ lower bound 
for checking the existence of 
a memoryless winning strategy for $2$-player, partial-information, non-stochastic games
was shown by \citeauthor{ABPSTT2021}\cite{ABPSTT2021}.
Our result uses stochasticity but works for any bounded-memory strategies. Moreover, the $\NP$-hardness holds already for the $1$-player case.

As a second and main contribution,
we propose a framework
that yields simple approximation algorithms
low  in the polynomial hierarchy (at level $\le 2$).
These ultimately work by guessing polynomial representations of witnessing strategies
and then approximating outcomes in the resulting Markov chains. 
Inspired by \cite{FMB2013},
we use fixed precision floating-point notations
to polynomially represent and manipulate doubly-exponentially small values in probability distributions.

The key new ingredient is an iterative state-space reduction scheme for Markov chains.
By way of encoding this in $\FOR$, the first-order theory of the reals\cite{BPR2006}, we derive
(1) a $\PSPACE$ algorithm for checking the existence, and computing representations of, $\eps$-equilibria for $\eps\ge0$;
(2) that if suitable equilibria/strategies exist for $\eps>0$, then also ones that are  
polynomially representable and verifiable; %
(3) that the loss of precision incurred by implementing the reduction scheme approximately can be bounded.
The latter two points use a fixed-precision floating-point representation allowing us to represent doubly-exponentially small values in %
strategies using a polynomial number of bits.

Our results have direct consequences for solving several well-known types of 
games 
including partial-information parity and mean-payoff games \cite{HPPR2018},
stay-in-a-set games \cite{SS2002},
quitting games \cite{SV2001} 
and concurrent reachability, B\"uchi and parity games \cite{BBL2021,BORV2021, BBL2022,BBL2023}.
In particular, we show how to approximate values/equilibria for $2$-player, zero-sum concurrent mean-payoff games (not discounted and with no assumptions on strategies - in particular, allowing infinite memory ones as is sometimes required) \cite{MN1981,KMH2008,HKM2009,HKMT2011,O2021}. Our algorithm is in \FNPNP, improving on the state-of-the-art algorithms that use polynomial space based on reductions to the existential theory of the reals \cite{HKMT2011,EY2008}.

\mypar{Paper structure}
We start by introducing notations in \cref{sec:preliminaries}
before stating our results formally in \cref{sec:results}.
In \cref{sec:lower-bound} we give $\NP$- and $\coNP$-lower bounds.
\Cref{sec:state-space-reductions} presents a value-preserving state-space reduction procedure for Markov Chains that forms the basis for our upper bounds.

In \cref{sec:FOR},
we show how to express games and questions about strategies,
in the first-order theory of the reals
and derive $\PSPACE$ upper bounds for checking the existence of $\eps$-equilibria 
and good zero-sum strategies.
In \cref{sec:stationary} we present our main approximation scheme
and the resulting upper bounds.
This relies on an encoding into \FOR\ to derive polynomial bounds on the representations of witnessing strategies 
and an argument that small perturbations of strategies lead to only small changes in values (\cref{ssec:poly-witnesses}). 
In \cref{ssec:approx-mc} we show how to approximate values in
the Markov chains resulting from fixing candidate strategy profiles.
In \cref{ssec:approx-games} we tie the above together into an algorithm to approximate values or equilibria in \FNPNP.
We conclude by discussing
applications and new results for related classes of games
in \cref{sec:applications}.

\section{Notations}
\label{sec:preliminaries}
For a set $X$, we write $X^k$ for its $k$-fold Cartesian product and $a[i]$ for the $i$th component of a vector $a\in X^k$.
Let $X^{*}$ and $X^{\omega}$ denote
the sets of finite and infinite sequences with elements in $X$, respectively.
We write 
$\dist(X)$ for the set of all probability distributions over $X$.

\mypar{Partial-information Games}

A $k$-player game is given by finite sets of states $\states$, actions $\actions$, signals $\signals$, rewards $\colours$, and a transition function $\trans:\states \times \actions^k \to \dist(\colours^k\x\signals^k\x\states)$ that maps the current state and action vector to a probability distribution over triples describing rewards, next signals, and the next state.
The game is played in stages $i\in \{0,1,2,\dots\}$, in each of which every player receives a signal, selects an action and receives a reward.
At stage $i$, the game is in a vertex $\state_i$.
Each player $j$ is sent the signal $\signal_{i}[j]$
and selects an action $\action_{i}[j]$.
Then a triple of rewards, signals, and successor state $(\colour_{i+1},\signal_{i+1},\state_{i+1})$ is drawn randomly according to $\trans(\state_i,\action_i)$
and the game moves to the next stage in vertex $\state_{i+1}$.
This random process goes on forever and describes
an infinite sequence
$\rho=(\state_0, \signal_0,\action_0, \colour_0),(\state_1,\signal_1,\action_1,\colour_1)\dots
\in (\states\x\signals^k\x\actions^k\x\colours^k)^{\omega}$.
We call such an infinite sequence a \emph{play} and let $\plays$
denote the set of all plays.

Note that games as described above are \emph{partial-} (aka.~\emph{imperfect}) information games:
the players do not directly observe the current state, the actions selected by other players nor the rewards obtained (not even their own). All they get to see is the signals sent to them.
We call a game \emph{perfect-information} if, at every stage $i>0$, every player receives as a signal the full action vector $\action_{i-1}$, reward vector $\colour_{i-1}$ and the current state $\state_i$.
A game is \emph{turn-based} (as opposed to \emph{concurrent}) if in any stage, only one player's action influences the next state and signals. That is, for every vertex $\state\in\states$ there is some player $j$ so that $\trans(\state,\action) = \trans(\state,\action')$ whenever $\action[j]=\action'[j]$.

\mypar{Plays and Strategies}
\label{ssec:defs:strategies}
\newcommand{\stratAct}[1][\strategy]{{#1}_{\mathop{act}\nolimits}}
\newcommand{\stratUp}[1] [\strategy]{{#1}_{\mathop{up}\nolimits}}

A \emph{strategy} %
is a function
$\strat:\signals^{*} \to \dist(\actions)$,
that determines, for every sequence of signals, a probability distribution over the action set. 
It is based on a set $M$ of memory modes if it can be described by an
initial memory mode $m_0\in M$ and a
pair of functions
$\stratAct: M\x S \rightarrow \dist(\actions)$ and
$\stratUp: M \x S \rightarrow \dist(M)$
that select actions and update the memory mode, respectively.
That is, for any signal sequence, 
$s=s_0s_1s_2\ldots s_j\in \signals^*$,
$\sigma(s) = \stratAct(m_j,s_j)$ and
$m_{j+1}\eqdef \stratUp(m_j,s_j)$.

A strategy is called \emph{finite-memory} if it is based on a finite memory set $M$
and \emph{memoryless} if it is based on a singleton set $M$.
A strategy is \emph{discrete} (also \emph{pure}) if all its choices are Dirac.
In perfect-information games, we call a strategy 
\emph{stationary} if 
its choices only depend on the current state (and not previous actions or rewards).

\mypar{Strategy profiles, Probability measures}
A \emph{strategy profile} is a family $\sigma=(\sigma[j])_{j\le k}$ of strategies, one for each player.
Together with an initial state $v_0$ and signal $s_0$,
it uniquely induces a probability measure
$\Prob{\sigma}{\state_0,\signal_0} \in \dist(\plays)$
over infinite plays
(see \cite{billingsley-1995-probability} for details). %
We will write $\ExpOf{\sigma}{\state_0,\signal_0}{X}$ for the derived expected value of a random variable $X$ and drop indices when clear from the context.
\mypar{Objectives}
Each player $j\le k$ has an objective, a measurable function $\obj_j: \plays \to \R$
and selects their strategy to maximise the expectation of this objective, against any possible opponent strategies.
A game is \emph{zero-sum} if $\sum_{j\le k} \obj_j(\rho)=0$ for every $\rho\in\plays$.

In this paper, we consider mean-payoff objectives as well as the special case of reachability objectives.
A \emph{mean-payoff} objective
(for player $j\le k$) assigns play
$\rho=(\state_0, \signal_0,\action_0, \colour_0),(\state_1,\signal_1,\action_1,\colour_1)\dots$
the value
\[
    \LMP(\rho)=
    \liminf_{N\to\infty}\frac{1}{N}\sum_{i=0}^{N}\colour_i[j].
\]
The \emph{reachability} objective for a target set $T\subseteq \states$ of states
assigns play
$\rho$ as above the value $1$ if $\state_i\in T$ for some round $i\ge 0$ and $0$ otherwise.

\mypar{Optimality criteria} %
The player in a 1-player game $G$ plays to maximise the expectation of the given objective $\?O$.
Here,
the \emph{value} of the game is 
$
\valueof{G} \eqdef \sup_{\sigma \in \Sigma} \ExpOf{\sigma}{\state_0,\signal_0}{\?O}
$
and
for $\eps\ge 0$, a strategy $\sigma$ is called \emph{$\eps$-optimal} if
$\valueof{G} - \ExpOf{\sigma}{\state_0,\signal_0}{\?O} \le \eps$.

For 2-player zero-sum games $G$,
let $\Sigma$ and $\Pi$ denote the set of strategies for players 1 and 2, respectively.
Then the \emph{lower value} and \emph{upper value} are defined as
\(
\valueoflower{G} \eqdef \sup_{\sigma \in \Sigma} \inf_{\pi \in \Pi}
    \ExpOf{\sigma,\pi}{\state_0,\signal_0}{\?O}
\)
and 
\(
\valueofupper{G} \eqdef \inf_{\pi \in \Pi} \sup_{\sigma \in \Sigma} 
    \ExpOf{\sigma,\pi}{\state_0,\signal_0}{\?O}
\).
By definition, it holds that $\valueoflower{G} \le \valueofupper{G}$.
If the two are equal
then this quantity is called the
\emph{value} of the game, denoted by $\valueof{G}$.
For $\eps \ge 0$, a strategy $\sigma\in\Sigma$ is called \emph{$\eps$-optimal}
if
\(
\valueof{G} - \inf_{\pi \in \Pi} \ExpOf{\sigma,\pi}{\state_0,\signal_0}{\?O}\le \eps
\) (and similar for player~2 strategies $\pi$).

For general $k$-player games $G$
and  any $\epsilon\geq 0$, an \emph{$\epsilon$-Nash equilibrium} (NE) is a strategy profile $(\sigma_j)_{j\le k}$
so that
any unilateral deviation by a player $j$ increases their expected outcome by at most $\eps$.
That is, for any profile $\strategy'$ where $\strategy'[i]=\strategy[i]$ for all $i\neq j$,
it holds that
$\ExpOf{\strategy'}{}{\obj_j} \le \ExpOf{\strategy}{}{\obj_j}+\eps$.

\mypar{Complexity classes}
Function problems
are given by a relation $R\subseteq\{0,1\}^*\x\{0,1\}^*$
where for some polynomial $p$, every $(x,y)\in R$ implies $\abs{y}\le p(\abs{x})$.
For a class $\?X$ of decision problems,
$R$ is in the corresponding function class $\FX$
if and only if the language $\{(x,y)\mid R(x,y)\}$ is in $\?X$,
meaning the question if $R(x,y)$ is true for given $x$ and $y$ is decidable in $\?X$.
This corresponds to the search problem
where for given $x$, one can either
stop and output $y$ so that $R(x,y)$, or correctly
state that no such $y$ exists.
In particular, we are concerned with the classes
\FNPNP, %
at level $2$ of the polynomial hierarchy.

\mypar{Floating-point Representations of Distributions}
\label{ssec:floats}
For any number $u\in\N$, let $\mathbb{Q}(u)$ denote the set of non-negative dyadic rationals with a $u$-bit mantissa
\[
\mathbb{Q}(u)\eqdef \{0\} \cup \{x2^{-i} ~\mid~ x\in \{2^{u-1}, 2^{u-}+1,\dots, 2^{u}-1\},i\in \mathbb{Z}\}.
\]
The $u$-bit floating point representation of $x2^{-i}$ is
$\langle 1^u, \binexp{x}, \binexp{i}\rangle$, where $\binexp{\cdot}$ is the binary expansion.
The floating point representation of $0$ is $\langle 1^u,0\rangle$. 
The \emph{relative distance} between two non-negative 
real numbers $x$ 
and $x'$
as $\rdist(x,x')= \frac{\max(x,x')}{\min(x,x')}-1$ with the
convention that $0/0=1$ and $c/0=+\infty$ for any $c>0$.
Two non-negative reals $x$ and $x'$ are \emph{$(u,i)$-close}
if $\rdist(x,x')\leq \left(\frac{1}{1-2^{-u+1}}\right)^i-1$. 
Let $\+P(u)$ denote the set of finite probability distributions 
$(p_1,p_2,\ldots,p_n)$
so that there exist numbers
$(p'_1,p'_2,\ldots,p'_n)\in\+Q(u)$ with
$p_i=p'_i/\sum_jp'_j$ for all $i=1,\ldots,n$ and so that
$\sum_jp'_j$ is $(u,n)$-close to $1$.
We call the primed vector the ($u$-bit) \emph{\ANR} of the unprimed distribution.

\section{Overview of Our Results}
\label{sec:results}

We consider the problem of approximating $\eps$-Nash equilibria and zero-sum values
with respect to bounded-memory strategies.

\begin{definition}
Consider some $k$-player game
with objectives $(\?O_j)_{j\le k}$ and a set $\adstrat$
of admissible strategy profiles.
Let $\eps\ge 0$.

An $\epsilon$-Nash equilibrium
\emph{with respect to} $\adstrat$
is a strategy profile $(\sigma_j)_{j\le k} \in \adstrat$
so that
any unilateral deviation by a player $j$ 
that yields another profile in $\adstrat$
can at most increase their expected outcome by $\eps$.
That is, for any profile $\sigma'\in \adstrat$ where $\strategy'[i]=\strategy[i]$ for all $i\neq j$,
it holds that
$\ExpOf{\strategy'}{}{\obj_j} \le \ExpOf{\strategy}{}{\obj_j}+\eps$.

In the $2$-player zero-sum case, we say a strategy 
$\sigma$ for player $1$ \emph{$\epsilon$-achieves} value $v\in\R$ with respect to $\adstrat$
if it guarantees expected outcome of at least $v-\eps$ against all admissible opponent strategies.
That is,
if $\inf_{\pi ~\mid~ (\sigma,\pi) \in \adstrat} \ExpOf{}{\sigma,\pi}{\obj_1} \ge  v-\eps$.
\end{definition}

For instance, an $\eps$-Nash equilibrium {with respect to} finite-memory strategies, is a profile consisting of finite-memory strategies, so that
no player can, by switching to another finite-memory strategy, improve their expected payout by more than $\eps$.

\medskip
We are interested in approximating such equilibria and zero-sum strategies
with respect to finite memory strategies with bounded memory.
Formally, we are given a $k$-player partial-information game $G$, some $\eps\in [0,1]$ and a vector $\vec{b}\in \N^k$ of memory bounds, one for each player $i\le k$.
Write $\BSP{i}{\vec{b}}$ for the set of player~$i$ strategies that are based on a set $M$ of size $\card{M}\le \vec{b}[i]$
and let $\BSB{\vec{b}}$ be all strategy profiles wherein each player~$i$ uses a strategy in $\BSP{i}{\vec{b}[i]}$.
For multiplayer games ($k\ge 2$) we ask if there exists an $\eps$-Nash equilibrium with respect to (wrt) $\BSB{\vec{b}}$ and if so, we want to compute one.
For $2$-player zero-sum games, we ask if there exists a strategy that
$\epsilon$-achieves a given value $v$  with respect to (wrt) $\BSB{\vec{b}}$.

Throughout, we assume that the number $k$ of players is fixed;
The inputs $\eps$, $v$ and rewards are given in binary
and the bounds $\vec{b}$ in unary encoding\footnote{This assumption is natural as we want to compute representations of strategies, in particular, which
mixed action to play for any combination of game and memory state.}.
The probability distributions dictated by the transition function $\trans$
are assumed to be given in $u$-bit approximately normalised floating point notation.

\begin{theorem}
    \label{thm:main-1}
    \label{thm:main-existence}
    For every fixed $k\ge 1$ the following is in \PSPACE.
    \begin{enumerate}
    	\item Given a $k$-player game, bounds $\vec{b}\in\N^k$ and $\eps\ge 0$, 
    	check whether an $\eps$-Nash equilibrium
    	exists with respect to $\BSB{\vec{b}}$.
    	\item Given a $2$-player zero-sum game, bounds $\vec{b}\in\N^2$, a number $v\in\R$ and $\eps\ge 0$, check whether there exists a player $1$ strategy that $\eps$-achieves value $v$ 
    	with respect to $\BSB{\vec{b}}$.   
    \end{enumerate}
\end{theorem}

Our proof is by an effective polynomial reduction of the problem to
the satisfiability problem of $\FOR$ formula so that 
the size of the constructed formula is only polynomial in the input,
has fixed alternation depth of $2$,
and solutions dictate values and witnessing strategies/equilibria.
At its core, our encoding relies on a state reduction technique for finite Markov chains, that results from fixing a strategy profile. This requires only polynomially many collapses and can therefore be ``implemented'' as a
small system of polynomial equations.

\medskip
We show that equilibria, zero-sum values and witnessing strategies can be approximated, up to an exponential additive error, at level 2 of the polynomial hierarchy.
\begin{restatable}{theorem}{thmfnp}
    \label{thm:main-2}
    \label{thm:main-appr}
    For every fixed $k\ge 1$ the following are in \FNPNP.
    \begin{enumerate}
        \item 
            Given a $k$-player game, bounds $b\in\N^k$ and $\eps \in (0,1/2)$,
            find an $\epsilon$-Nash equilibrium with respect to $\BSB{\vec{b}}$
            (if for some $\eps'\in [0,\eps/3)$,
            an $\eps'$-Nash equilibrium with respect to $\BSB{\vec{b}}$ exists);
        \item 
            Given a $2$-player zero-sum game, bounds $b\in\N^k$, $\eps \in (0,1/2)$ and $v\in \R$, 
            find a strategy that $\eps$-achieves $v$ with respect to $\BSB{\vec{b}}$
            (if for some $\eps'\in [0,\eps/3)$,
            some strategy $\epsilon'$-achieves $v$ with respect to $\BSB{\vec{b}}$).
    \end{enumerate}
\end{restatable}

Using the second point and binary search one can $\eps$-approximate
the lower and upper values wrt $\BSB{\vec{b}}$
in $2$-player zero-sum games.
These complexity upper bounds above directly apply to many classes of games that admit memoryless $\eps$-optimal strategies ($\BSB{\vec{1}}$).

\medskip
In addition to the approximation algorithms, we provide an $\NP$- and $\co\NP$-lower bound for approximating zero-sum values for partial-information games.
This already holds for a very restricted class of games, as follows.
Let us call a game \emph{simple} if it is 
a 1-player reachability game so that 
the underlying transition graph is acyclic (except for target states);
the motion from the initial state $\state_0$ is uniformly at random over a set of successor states; and
for all actions $\action$ and states $\state\neq \state_0$
the distributions $\trans(\state,\action)$ are Dirac.

We show that it is hard to check if the value achieved by memoryless strategies exceeds a given threshold.
We state this in the form of a gap problem (cf.~\cite{G2006}).

\begin{theorem}
\label{thm:lower-bound}
The following decision problem is \NP-hard. %
Given $0\le l<u\le 1$ in unary encoding
and a simple game
in which
the zero-sum value is not in $(l,u)$
and in which there exists an optimal memoryless deterministic strategy.
Is the value greater than or equal to $u$?
\end{theorem}

This lower bound directly applies to threshold problems regarding the values of zero-sum games, approximating zero-sum values and witnessing strategies, and to approximating $\eps$-equilibria in multiplayer games, wrt memoryless strategies ($\BSB{\vec{1}}$).

\begin{figure*}
\subfigure[
$C_i$ contains $X_j$.
]{
\begin{tikzpicture}[
            node distance=0.75cm and 2cm,
    ]
	\node[] (I11) at (0,0) {$i,j$};
	\node[right=of I11] (I12) {$i,j+1$};

        \node[above=of I11] (W1) {$w,j$};
	\node[right=of W1] (W2) {$w,j+1$};

        \path (W1) edge  node[auto,sloped,lbl] {$\true,\false$} (W2);
        \path (I11) edge node[auto,sloped,lbl] {$\true$} (W2);
        \path (I11) edge node[auto,sloped,lbl] {$\false$} (I12);

\end{tikzpicture}
}
\hspace{1.5cm}
\subfigure[
$C_i$ contains $\lnot X_j$.
]{
\begin{tikzpicture}[
            node distance=0.75cm and 2cm,
    ]
	\node[] (I11) at (0,0) {$i,j$};
	\node[right=of I11] (I12) {$i,j+1$};

        \node[above=of I11] (W1) {$w,j$};
	\node[right=of W1] (W2) {$w,j+1$};

        \path (I11) edge node[auto,sloped,lbl] {$\true$} (I12);
        \path (W1) edge  node[auto,sloped,lbl] {$\true,\false$} (W2);
        \path (I11) edge node[auto,sloped,lbl] {$\false$} (W2);

\end{tikzpicture}
}
\hspace{1.5cm}
\subfigure[
$C_i$ contains neither $X_j$ nor $\lnot X_j$
]{
\begin{tikzpicture}[
            node distance=0.75cm and 2cm,
    ]
	\node[] (I11) at (0,0) {$i,j$};
	\node[right=of I11] (I12) {$i,j+1$};
        \node[above=of I11] (W1) {$w,j$};
	\node[right=of W1] (W2) {$w,j+1$};

        \path (W1) edge  node[auto,sloped,lbl] {$\true,\false$} (W2);
        \path (I11) edge node[auto,sloped,lbl] {$\true,\false$} (I12);

\end{tikzpicture}
}
\caption{Stages $0<j\le m$ in game $G_\varphi$. In step $j\leq m$ 
the player receives signal $j$ and
the pebble is moved out of some state $(*,j)$.
}
  \label{fig:lower-bound}
\end{figure*}
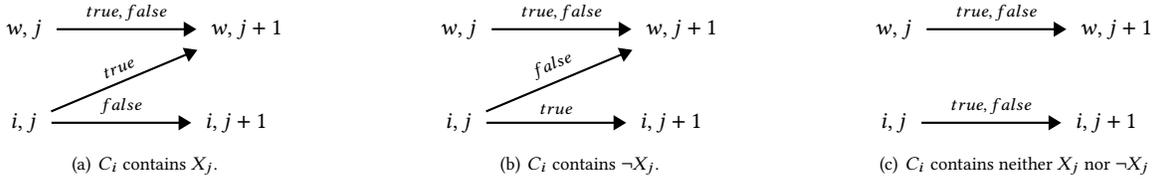

\section{Lower Bound}
\label{sec:lower-bound}
We prove \cref{thm:lower-bound} by reduction from the Boolean satisfiability problem:
an instance of 3SAT with $n$ clauses and $m$ variables 
variables $X_1,X_2,\ldots,X_m$
is a conjunction
\[
\varphi = \bigwedge_{i=1}^n \bigvee_{j=1}^3 L_{i,j}
\]
where each conjunct is called a clause and each 
$L_{i,j}$ is either a positive literal
(a variable in $\{X_k\mid k\le m\}$)
or a negative literal (the negation of a positive literal).
W.l.o.g., assume that no variable appears both positively and negatively in any clause.
An \emph{interpretation} is a mapping $\nu:\{1,\ldots,m\}\to \{\true,\false\}$ assigning a truth value to each variable.
The (\NP-hard) decision problem is if there exists an interpretation $\nu$ under which the given conjunction is true, meaning that for every clause $i$ there is some $j$ such that either
$L_{i,j}=X_k$ and $\nu(k)=\true$
or
$L_{i,j}=\neg X_k$ and $\nu(k)=\false$.
Analogously, it is \coNP-hard to check if, for a given 3SAT instance, no interpretation satisfies the given conjunction.

We proceed to describe a game $G_{\varphi}$ for a given instance $\varphi$.
This game is set up so that in the first round, independently of the player's action,
a random clause is selected to be checked.
In the following $m$ rounds, the player's actions determine an
interpretation, one variable per round and without knowing which clause is being checked.

Formally, for every $0\le j\le m+1$ the game has a state $(w,j)$
as well as states $(i,j)$ for all $0\le i\le n$.
The available actions are $\actions=\{\true,\false\}$. 
The game starts in $\state_0=(0,0)$ and in the first round, under both actions, moves uniformly at random into a state in $\{(i,1) \mid 0\le i\le n\}$, thereby selecting the clause $i$ to be checked.

In round $0<j\le m$, the player picks an action $\true$ or $\false$ to indicate that variable $X_j$ has that value in his chosen interpretation. 
From state $(w,j)$, the pebble moves surely to $(w,j+1)$, regardless of this choice.
Otherwise, from state $(i,j)\neq(w,j)$, there are three cases
depending on whether clause $C_i$ is satisfied by setting the value of $X_j$ as indicated by the player's action (cf.~\cref{fig:lower-bound}).

\begin{enumerate}[label=(\alph*)]
    \item 
If $C_i$ contains $X_j$ positively, then upon action $\true$ the pebble surely moves to $(w,j+1)$ and upon action $\false$, it surely moves to $(i,j+1)$.

\item If $C_i$ contains $X_j$ negatively, then the roles are reversed: upon action $\false$ the pebble surely moves to $(w,j+1)$ and upon action $\true$, it surely moves to $(i,j+1)$.

\item If $C_i$ does not contain $X_j$, then upon both actions the pebble surely
to $(i,j+1)$.
\end{enumerate}
In every round from a state $(*,j)$, the player receives signal $j$.
This ends the description of game $G_{\varphi}$.

The player only knows the current round number $j$ up to $m+1$, but not in which state the pebble resides, i.e., which clause is being checked and whether it is already satisfied. By his actions in rounds $1,\ldots,m$, he picks a distribution over the possible variable interpretations.

Let us write
$W\eqdef\XReach{\{(w,m+1)\}}$
for the set of plays that reach state $(w,m+1)$.
The following formalises the key property of the game $G_\varphi$.

\begin{lemma}
    \label{lem:lower-bound-vars}
Let $\nu$ be an interpretation that satisfies $b$ out of $n$ clauses
and let $\sigma$ be a strategy that plays action $\nu(X_j)$ in round $0<j\le m$.
Then $\Prob{\sigma}{\state_0,0}(W) = b/n$.
\end{lemma}
\begin{proof}
In the first round, nature picks uniformly at random a state $(i,1)$ for some $0<i\le n$.
Then in exactly $b$ out of $n$ many cases (the ones where clause $C_i$ is satisfied by the assignment), the play leads to $(w,m+1)$.
\end{proof}

\begin{proof}[Proof of \cref{thm:lower-bound}]
To show the \NP-hardness, we reduce from 3SAT.
Given an instance $\varphi$ with $n$ clauses and $m$ variables,
let $c$ be the maximal number of clauses satisfied by any interpretation.

Let $G$ be the game $G_\varphi$ where the player receives reward $1$ whenever the pebble is in state $(w,m+1)$ and $0$ otherwise.
This game is \emph{simple}; in particular, it is a 1-player reachability game with target $(w,m+1)$, encoded as a mean-payoff game.

We show that the value in $G$ is $1$ if $\varphi$ is satisfiable and at most $1-1/n$ otherwise. 
Suppose first that $\varphi$ is satisfiable. Then there exists an interpretation that satisfies $c=n$ clauses and the value is $1$, by \cref{lem:lower-bound-vars}.
Otherwise, if $\varphi$ is not satisfiable,
every interpretation satisfies at most $c$ clauses and thus witnesses at most 
value $c/n$.
As the history of signals is the same for all infinite plays,
every (arbitrary) strategy
$\sigma$ defines a distribution over all possible variable interpretations.
The value it witnesses is the linear combination
of the values witnessed by discrete strategies, and so also at most
$c/n\le 1-1/n$. 

Let $l\eqdef 1-1/n$ and $u\eqdef 1$.
By the argument above, the value of the game is not in $(l,u)$
and $\ge u$ if, and only if, $\varphi$ is satisfiable.
\end{proof}

By the argument above, it is \coNP-hard to check if the value of a simple game is less than or equal to a threshold $l$ given in unary.
Similarly, one can introduce a passive second player that receives inverted rewards ($-r$ whenever player~1 gets $r$).
This shows that it is \coNP-hard to check if the value of a $2$-player zero-sum safety game exceeds a threshold $u$.

The games constructed above all admit memoryless optimal strategies.
Therefore, approximating the value wrt memoryless ($\BSB{b=\vec{1}}$) strategies is at best \FNP, even for simple games.

\section{Mean-Payoff Values for Markov Chains}
\label{sec:alg-mc}
\label{sec:state-space-reductions}
We demonstrate how to iteratively collapse Markov chains
with the aim of computing the expected mean-payoff value of a fixed initial state. %
The underlying idea is to simultaneously summarise the expected reward and duration (number of steps) of paths between two nodes
and adjust these summaries consistently during collapses. Once an irreducible chain is produced it is trivial to read off the mean-payoff values.

A Markov chain is a $1$-player game with only a single, unique action $a\in \actions$ for the player in every state 
(the player is just there to collect the rewards and cannot influence the outcome).
To simplify notations, we will assume that a Markov chain $M$
has states $\states=\{1,\dots,n\}$
and for every edge from states $i$ to $j$, let
$\probm(i,j)$ and
$\rew(i,j)\in\R$ 
be the associated probability and reward respectively.
Each edge from $i$ to $j$ also has a duration denoted by 
$\dur(i,j)$. %
We define the \emph{\mpr} of an infinite play $\rho=s_0s_1\ldots$ as the limit of total reward by the total duration of expanding prefixes:

\begin{equation}
\label{def:mpr}    
\MPR(\rho) \eqdef \liminf_{N\to\infty}
            \frac{
                \sum_{i=0}^{N-1} \rew(s_i,s_{i+1})
            }{
                \sum_{i=0}^{N-1} \dur(s_i,s_{i+1})
            }
\end{equation}

Notice that this coincides with the usual \mpo\ $\LMP(\play)$ if all edges have duration $1$.

Let's recall a few notions about Markov chains.
Write $i\to j$ to denote that there is a path from state $i$ to $j$ with non-zero probability.
We say $i$ and $j$ \emph{communicate} if both $i\to j$ and $j\to i$.
A set $S\subseteq\states$ of states is communicating if all its elements are pairwise communicating.
$S$ is closed if $i\to j$ and $i\in S$ implies $j\in S$.
Closed communicating sets (CCSs) are also called {bottom strongly connected components}.
A state $i$ is \emph{recurrent} if it belongs to a CCS and \emph{transient} otherwise.
It is \emph{absorbing} if $\probm(i,i)=1$. That is, an absorbing state describes its own (singleton) CSS.

\medskip
We introduce two operations on Markov chains that will preserve the expected \mpr\
$\ExpOf{}{s}{\MPR}$ from any state $s$.
Proofs of their correctness (\cref{lem:state-elim,lem:loop-elim}) are in \cref{app:MC}.

\begin{definition}[{State Elimination}]
\label{def:state-elim}
Let $M=(\states,\probm)$ be a Markov chain with states $\{1,\dots,n\}$ and probability matrix $\probm$ so that $\probm(n,n)=0$
(state $n$ does not have a self-loop).
The Markov chain $M'=(\states',\probm')$ {results from $M$ by eliminating state $n$}
if it has states $\states'=\states \setminus \{n\}$ and,
for every two states $i,j\neq n$, the values of probabilities, rewards, and durations are updated as follows
(see \cref{fig:state_elim_small}).
\begin{align*}
		\probm'(i,j) &~\eqdef~ \probm(i,j) + \probm(i,n)\cdot \probm(n,j)\\
		\rew'(i,j)   &~\eqdef~ \probm(i,j)\rew(i,j) + (\probm(i,n)\cdot \probm(n,j))(\rew(i,n)+\rew(n,j))\\
		\dur'(i,j)   &~\eqdef~ \probm(i,j)\dur(i,j) + (\probm(i,n)\cdot \probm(n,j))(\dur(i,n)+\dur(n,j))		
\end{align*}
\end{definition}

\begin{figure}[t]
\centering
\begin{tikzpicture}[
            node distance=0.75cm and 1cm,
    ]
	\node[state] (4) at (0,0) {$i$};
	\node[state,below left=of 4] (2) {$n$};
	\node[state, below right=of 2] (3) {$j$};
	
        \path (4) edge
            node[lbl,swap,pos=0.2] {$\probm(n,i)$} 
            node[lbl,swap,pos=0.5] {$\rew(n,i)$} 
            node[lbl,swap,pos=0.8] {$\dur(n,i)$}
        (2);

        \path (4) edge 
            node[lbl,pos=0.3] {$\probm(n,j)$}
            node[lbl,pos=0.5] {$\rew(n,j)$}
            node[lbl,pos=0.7] {$\dur(n,j)$}
        (3);
        
        \path (2) edge 
            node[lbl,swap,pos=0.2] {$\probm(n,j)$}
            node[lbl,swap,pos=0.5] {$\rew(n,j)$}
            node[lbl,swap,pos=0.8] {$\dur(n,j)$}
        (3);
	
        \node (mid) at (2.25,-1) {$\leadsto$};
	
	\node[state] (4') at (4,0) {$i$};
        \node[state,draw=none,below left=of 4'] (2') {};
	\node[state, below right=of 2'] (3') {$j$};
	
        \path (4') edge 
            node[lbl,pos=0.3] {$\probm'(n,j)$}
            node[lbl,pos=0.5] {$\rew'  (n,j)$}
            node[lbl,pos=0.7] {$\dur'  (n,j)$}
        (3');
\end{tikzpicture}
\caption{Elimination of state $n$:
every length-two path via state $n$ in $M$ (on the left) is removed and the corresponding direct edge re-weighted in $M'$ (right).
The expected reward and duration of going from $i$ to $j$ remains the same.}
\label{fig:state_elim_small}
\end{figure}
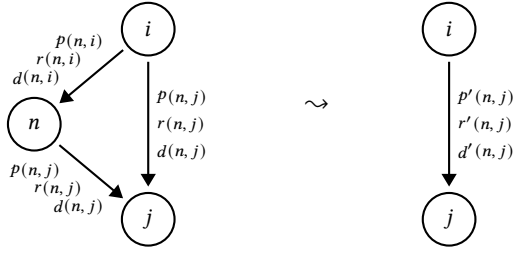

\begin{restatable}{lemma}{stateelim}
	\label{lem:state-elim}
	Suppose $M'$ results from $M$ by eliminating state $n$.
        Then 
        $\ExpOf{M}{s}{\MPR}=\ExpOf{M'}{s}{\MPR}$
        for every state $s\neq n$ .
\end{restatable}

\begin{figure}[t]
\centering
\begin{tikzpicture}[
            node distance=1.5cm and 0.4cm,
    ]
	\node[state] (2) at (0,0) {$n$};
	\node[state,below left=of 2] (4) {j};
	\node[state, below right=of 2] (3) {$i$};
	
        \draw (2) edge[loop left,looseness=15] node[lbl, align=left] {
            $\probm(n,n)$\\
            $\rew(n,n)$\\
            $\dur(n,n)$
        } (2);
        \path (2) edge
            node[lbl,pos=0.5] {$\probm(n,i)$} 
            node[lbl,pos=0.7] {$\rew(n,i)$} 
            node[lbl,pos=0.9] {$\dur(n,i)$}
        (3);
        
        \path (2) edge 
            node[lbl,swap,pos=0.5] {$\probm(n,j)$}
            node[lbl,swap,pos=0.7] {$\rew(n,j)$}
            node[lbl,swap,pos=0.9] {$\dur(n,j)$}
        (4);
	
        \node (mid) at (2.10,0) {$\leadsto$};
	
	\node[state] (30) at (4.5,0) {$n$};
	\node[state,below left=of 30] (33) {$j$};
	\node[state, below right=of 30] (32) {$i$};
	
        \path (30) edge 
            node[lbl,pos=0.5] {$\probm'(n,i)$}
            node[lbl,pos=0.7] {$\rew'  (n,i)$}
            node[lbl,pos=0.9] {$\dur'  (n,i)$}
        (32);
        
        \path (30) edge 
            node[lbl,swap,pos=0.5] {$\probm'(n,j)$}
            node[lbl,swap,pos=0.7] {$\rew'  (n,j)$}
            node[lbl,swap,pos=0.9] {$\dur'  (n,j)$}
        (33);
\end{tikzpicture}

\caption{Elimination of loop $n$ in $M$ (left).
In $M'$ (right), the corresponding edge has probability $0$ and all other edges from $n$
are re-weighted to match the expected duration and expected reward of paths that start by iterating
the loop.
}
\label{fig:loop_elim_small}
\end{figure}
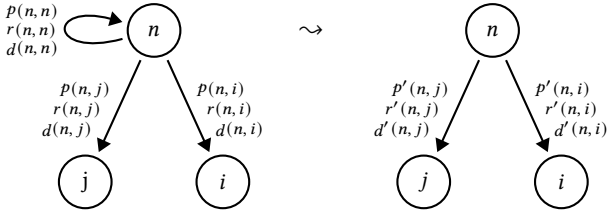

\begin{definition}[{Loop Elimination}]
\label{def:loop-elim}
Let $M=(\states,\probm)$ be a Markov chain with states $\{1,\dots,n\}$ and probability matrix $\probm$ so that $0<\probm(n,n)<1$
(state $n$ is not absorbing and has a self-loop).
Eliminating loop $n$ results in a Markov chain
$M'=(\states',\probm')$ with $\states'=\states$
where for any edge from state $i\neq n$ to $j$,
the probability, reward, and duration are as in $M$.
For edges going out of state $n$
they are $\probm'(n,n)=\rew'(n,n)=\dur'(n,n)=0$ and for all $j\neq n$,
	\begin{align*}
	\probm'(n,j) &~\eqdef~ \frac{\probm(n,j)}{1-\probm(n,n)} \\
	\rew'(n,j) &~\eqdef~ \frac{1}{1-\probm(n,n)}\rew(n,n) + \frac{\probm(n,j)}{1-\probm(n,n)} \rew(n,j) 		\\
	\dur'(n,j) &~\eqdef~ \frac{1}{1-\probm(n,n)}\dur(n,n) + \frac{\probm(n,j)}{1-\probm(n,n)} \dur(n,j) .	
	\end{align*}
\end{definition}
The intuition behind this definition is that iterating the loop $n\to n$ and then going to $j\neq n$ is replaced by a direct step that preserves the probability of ending in that particular state $j$,
as well as the expected sum of rewards and durations on the way.
Recall that the expected number of times the loop is used is ${1}/{1-\probm(n,n)}$.

\begin{restatable}{lemma}{loopelim}
	\label{lem:loop-elim}
	Suppose $M'$ results from $M$ by eliminating loop $n$.
        Then 
        $\ExpOf{M}{s}{\MPR}=\ExpOf{M'}{s}{\MPR}$
        for every state $s\le n$ .
\end{restatable}

The two elimination operations 
can be used to determine \mpr values in finite Markov chains.
Indeed, suppose we want to determine $\ExpOf{M}{s}{\MPR}$. %
Exhaustively eliminating all loops and states except $s$ results
in a Markov chain $M'$ where every original CCS is collapsed into a single absorbing state.
By \cref{lem:state-elim,lem:loop-elim},
the expected \mpr\ of state $s$ is the same in $M$ and $M'$.
Now, if $M'$ contains only a single absorbing state $s$, then trivially
$\ExpOf{M'}{s}{\MPR} = \rew'(s,s) / \dur'(s,s)$.
Otherwise, if $s$ is the only transient state in $M'$
then $\ExpOf{M'}{s}{\MPR} =\sum_{i\neq s} \probm'(s,i) \left(\rew'(i,i) / \dur'(i,i)\right)$.

Notice that the number of elimination steps is at most $2\abs{\states}$.
We will implement this procedure %
symbolically in \FOR\ in the next section,
and approximately, in \cref{sec:approx}, to prove our main results.

\section{Expressibility in \FOR}
\label{sec:FOR}
We fix a $k$-player game $G$, bounds  $\vec{b}\in \N^k$ and let 
\[N\eqdef \card{\states}\cdot\card{\signals}^k\cdot\prod_{i=1}^k \vec{b}[i].\]
We show how to  express 
$\eps$-Nash equilibria wrt $\BSB{\vec{b}}$ (in the multiplayer case) 
and
strategies that $\epsilon$-achieve some value $v$ wrt $\BSB{\vec{b}}$ (in the $2$-player zero-sum case)
in $\FOR$. 
The steps to do this are as follows. We will elaborate on each step in subsequent paragraphs.
\begin{enumerate}	
\item We existentially guess 
a strategy profile %
(in the multiplayer case), 
or 
a strategy for player $1$
(in the $2$-player, zero-sum case),
such that the probability distributions used in the strategies are encoded in the variables. 
\item For any player $i$, we have a formula that depends on variables encoding a strategy profile, and the formula encodes the outcome for the player $i$ in the Markov chain induced by the strategy profile.
\item We encode that the strategies guessed in step $1$
satisfy the desired criteria. 
\end{enumerate}

\paragraph{Step 1: Guessing strategies and strategy profiles\label{par:guess}}
Encoding a strategy in $\BSP{i}{\vec{b}}$ for any player $i$ 
(and thus strategy profiles in $\BSB{\vec{b}}$), 
in $\FOR$ is straightforward. 
Indeed, a strategy in $\BSP{i}{\vec{b}}$ consists of 
an action function $\stratAct: M\x S \rightarrow \dist(\actions)$ and 
a memory update function $\stratUp: M \x S \rightarrow \dist(M)$.
For the action function, we have for each signal, each memory state and each of the player's possible actions under that signal, 
a variable. 
Specifically, if the actions with signal $s$ and memory $m$ are $a_1,\dots,a_{\card{\actions}}$, 
we have variables $q_{1}^{s,m,i},\dots,q_{\card{\actions}}^{s,m,i}$. 
To express that they should form a probability distribution, 
we have inequalities $0\leq q_{j}^{s,m,i}\leq 1$, for each $j\in\{1,\dots,\card{\actions}\}$ and 
$\sum_{j=1}^{\card{\actions}}q_{j}^{s,m,i}= 1$. 

Similarly, for the memory update function $\stratUp$, we just, for each signal $s$, each memory $m$, and each action vector $a$ in that state have variables $b_{1}^{s,m,a,i},\dots,b_{\vec{b}[i]}^{s,m,a,i}$. 
To express that they should form a probability distribution, 
we have inequalities $0\leq b_{j}^{s,m,a,i}\leq 1$ for each $j\in\{1,\dots,\vec{b}[i]\}$ and $\sum_{j=1}^{\vec{b}[i]}b_{j}^{s,m,a,i}= 1$. 
This gives us the following lemma.

\begin{lemma}
	\label{lem:exFO-strat}
	For any $k$-player game $G$ and $i\le k$, %
        one can construct a $\FOR$ formula $\text{Strat}_i(\vec{x_i})$, with $|\vec{x_i}|= (\card{\signals}\cdot\vec{b}[i])(\card{\actions}+\vec{b}[i])$ variables, 
	such that there is a one-to-one correspondence between satisfying assignments of  $\vec{x_i}$
	 and strategies in $\BSP{i}{\vec{b}}$
        for player~$i$.
\end{lemma}

\paragraph*{Step 2: Evaluating the induced Markov chain}
Once the strategy profile is fixed, we obtain an induced Markov chain. 
The following lemma shows that there is an $\FOR$ formula that uses variables to identify a strategy profile and 
encodes the value obtained by the strategy profile  
by encoding the loop and state elimination procedure
from \cref{sec:state-space-reductions}. 
The proof, showing how this formula is constructed is in \cref{app:exFOR}.
In essence, the loop/state elimination procedures are iteratively implemented in FO(R) by introducing lots (but polynomially many) existentially quantified variables to represent intermediate Markov chains in the reductions.

\begin{restatable}{lemma}{lemexFO}
\label{lem:exFO}
For any $k$-player game $G$, and a player $i$,
one can construct
a formula $\text{Val}_i(\vec{x}_1,\dots, \vec{x}_k, y)$, 
such that
$\vec{x}_j$ encodes strategy $\sigma_j$ of a strategy profile $\sigma\in \BSB{\vec{b}}$, 
and the formula uses
a set of Boolean combination of $100N^3+40N$ polynomials 
with at most $8N^3+3N$ variables, each of degree at most $2k$, 
and coefficients coming from $G$ 
and $y$ encode the outcome of $\sigma$ for player $i$.
\end{restatable}

\paragraph*{Step 3: Finding $\epsilon$-achieving strategy/$\epsilon$-Nash equilibria}
The last step requires unfolding the definition of $\epsilon$-achieving strategy/$\epsilon$-Nash equilibria in a formula. 
We first do this for the $2$-player zero-sum case where we are interested in strategies that $\epsilon$-achieves a given value $v$, for $\epsilon\ge 0$.
The formula guesses the $\epsilon$-achieving strategy, and checks for 
each of the opponent's strategies, that the outcome is $\ge v-\epsilon$.  
\begin{equation*}
\begin{split}
         \exists \vec{x_1} \forall \vec{x_2} \forall y.
         \quad
        &\text{Strat}_1(\vec{x_1}) \land \text{Strat}_2(\vec{x_2}) \land
        \text{Val}_1(\vec{x_1},\vec{x_2},y)\\
        &\implies (y\ge v -\epsilon)
\end{split}
\end{equation*}
Note that one can use existentially quantified auxiliary variables to implement repeated squaring and thereby encode the (binary) numbers $v$ and $\eps$ succinctly.

For the case of $\eps$-Nash equilibria wrt $\BSB{\vec{b}}$, 
we guess the strategy profile $\sigma$ and check against all possible deviations by a single player. 
\begin{equation*}
	\begin{split}
                \exists &\vec{x_1},\dots , \vec{x_k}, y_1,\dots ,y_k
          \forall \vec{x'_1},\dots , \vec{x'_k},y'.\\
		 &\left(\bigwedge_{i=1}^k \text{Strat}_i(\vec{x_i})\right)  \wedge \left(\bigwedge_{i=1}^k \text{Val}_i(\vec{x_1},\dots, \vec{x_k},y_i)\right) \wedge\\
		& \bigwedge_{j=1}^k \left(\text{Strat}_j(\vec{x'_j}) \implies
		(\text{Val}_j(\vec{x_1},\dots,\vec{x'_j},\dots, \vec{x_k},y') \wedge y_j+\epsilon \ge y')\right)
	  \end{split}
\end{equation*}

\begin{restatable}{lemma}{lemexFOeq}
	\label{lem:exFO-eq}
Given a $k$-player game, and $v$, $\epsilon\ge 0$, 
one can construct an \FOR\ formula
that encodes a strategy that $\eps$-achieves $v$/ is a $\eps$-Nash equilibrium. Furthermore, the formula uses at most
$(k+1)+k(|S|\prod_{i=1}^k\vec{b}[i])(|A|+\sum_{i=1}^{k}\vec{b}[i])+ 8N^3+3N$ variables,
and at most $(100N^3+40N)2k$ polynomials.
\end{restatable}

The terms in the number of variables correspond to variables for guessing 
values, strategies and the variables used to describe the values in \cref{lem:exFO}.

\begin{proof}[Proof of \cref{thm:main-1}]
	From \cref{lem:exFO-eq}, we can express strategies encoding an $\eps$-Nash equilibrium using an $\FOR$ formula. Note that this uses $2$ alternations of quantifiers. 
	Checking if a formula with bounded quantifier alternation 
	is in the first-order theory of the reals can be done in $\PSPACE$ as shown by \citeauthor{BPR2006}\cite[Remark 13.10]{BPR2006}.
\end{proof}

The following will be used later applied to the formulae provided by \cref{lem:exFO-eq}. For the proof, we assume familiarity with \FOR, the first order theory of the reals \cite{BPR2006}.
\begin{restatable}{lemma}{thmFORbounds}
\label{thm:FOR}
Consider a satisfiable \FOR\ expression, 
\[\exists x_1^1,\dots,x_{n_1}^1\forall x_1^2,\dots,x_{n_2}^2\exists\dots x_1^k\dots,x_{n_k}^k. \varphi(x_1^1,\dots,x_{n_1}^1,x_1^2\dots,x_{n_k}^k)\]

where $\varphi$ consists of boolean combinations over 
$s$ polynomial inequalities, 
where each polynomial is of degree at most $d$, has integer coefficients and the bit-size of each coefficient is at most $\tau$.
Then there exists $x_1',\dots,x_{n_1}'$, such that, 

$\forall x_1^2,\dots,x_{n_2}^2\exists\dots \exists x_1^k\dots,x_{n_k}^k. \varphi(x_1',\dots,x_{n_1}',x_1^2\dots,x_{n_k}^k)$

is satisfiable and such that for each $i$,  either $2^{-\tau (2d)^{\prod_{i=1}^k O(n_i)}}\leq |x_i'|\leq 2^{\tau (2d)^{\prod_{i=1}^k O(n_i)}}$ or $x_i'=0$.
\end{restatable}

\begin{proof}
	We apply 
	\cite[Thm~14.16]{BPR2006}
	to 
	\begin{align*}
		\phi(x_1^1,\dots,x_{n_1}^1)\eqdef& 
	\forall x_1^2,\dots,x_{n_2}^2\\
	&\exists\dots x_1^k\dots,x_{n_k}^k: \varphi(x_1^1,\dots,x_{n_1}^1,x_1^2\dots,x_{n_k}^k)
	\end{align*}
	which gives us a quantifier-free expression of  $\phi$. This expression of $\phi$ uses polynomials of degree at most $d'\eqdef d^{\prod_{i=2}^k O(n_i)}$ with integer coefficients, each of bit-size at most $\tau'\eqdef \tau d^{\prod_{i=1}^k O(n_i)}$.
	We then apply 
	\cite[Thm~13.10]{BPR2006}
	to $\phi$, giving us a ``univariate representation'' of each of the numbers we will use for $x_i'$ 
	(the theorem itself only states that we get approximations, but the complexity analysis of the algorithm states that it is equally fast to get exact solutions, using some additional steps). 
	For each $i$, this univariate representation for $x_i$ consists of a fraction of two polynomials $p_i(x)/q(x)$ and an additional polynomial $f$ ($q$ and $f$ are coprime). The number $x_i$ is then $p_i(t)/q(t)$ where $t$ is a root of $f$. These polynomials each is of degree at most $d''\eqdef (2d'+6)^{n_1}\leq (2d)^{\prod_{i=1}^k O(n_i)}$ (we get to ignore the $6$ because of the $O$'s. We can't ignore the 2 in case $d$ was 1) with integer coefficients of bit-size at most $\tau''\eqdef \tau d^{\prod_{i=1}^k O(n_i)}$. 
	For each $i$, we apply
	\cite[Lem.~15]{BPR2006}
	to $f$, $p_i$ and $q$, giving us a square-free univariate polynomial $p'_i$ such that $p'_i(x_i)=0$. This polynomial has degree at most $d'''\eqdef 2d''=(2d)^{\prod_{i=1}^k O(n_i)}$ and the bitsize representation of the integer coefficients is at most $\tau'''\eqdef 2d''\tau''+7d''\log d''=\tau (2d)^{\prod_{i=1}^k O(n_i)}$.
	
	For any univariate polynomial $p(x)=a_n x^{n-1}+a_{n-1} x^{n-2}+\dots+a_1 x+a_0$, then $r\neq 0$ is a root iff $1/r\neq 0$ is a root of $a_n 1/(x^{n-1})+a_{n-1} (1/x^{n-2})+\dots+a_1 1/x+a_0$ or equally $p'(x)\eqdef a_n +a_{n-1} x+\dots+a_1 x^{n-2}+a_0 x^{n-1}$. Therefore, if $U$ is an upper bound on the absolute largest root in $p$, then $1/U$ is a lower bound for how small the absolute smallest non-zero root can be in $p'$. 
	We can bound $U\le 1+\max\left(\abs{\frac{a_{n-1}}{a_{n}}},\abs{\frac{a_{n-2}}{a_{n}}},\dots \abs{\frac{a_{0}}{a_{n}}}\right)$
	\cite{Cau1828}
	and so
	get that $2^{-\tau'''}\leq |x_i'|\leq 2^{\tau'''}$, because the largest integer represented with $x$ bits is $2^{x}-1$.
\end{proof}

\section{Approximation Algorithms}
\label{sec:stationary}
In this section, we present the proof of %
\cref{thm:main-2}. 
Our approach is to repeat the steps leading to \FOR\ representations in the previous section but use floating point representations and associated lossy arithmetic.
To do this, we need two ingredients:
(1) we need to guess candidate profiles in polynomial space;
and (2) we need to compute approximations of the mean-payoff values in the induced Markov chains in polynomial time.

We focus first on the existence of small candidate profiles.
\subsection{Floating point representations}

\begin{definition}
A strategy profile
$\sigma=(\sigma[j])_{j\le k}
\in \BSB{\vec{b}}$
is \emph{$u$-bit representable} if all its distributions
(selecting mixed actions and memory updates for all players, memory modes and signals)
are in $\+P(u)$.
The profile $\sigma$ is \emph{represented} in $u$-bits if it is $u$-bit representable
and all exponents are written in binary using $u$ many bits.

Let
\begin{align*}
	p(\sigma,\vec{m},s,\vec{m'},\vec{a})\eqdef
	\prod_{j=1}^k&\stratAct[{\sigma[j]}](\vec{m}[j],s[j])(\vec{a}[j])\cdot\\
	&\stratUp[{\sigma[j]}](\vec{m}[j],s[j])(m'[j]) 
\end{align*}
denote the probability 
of that starting with memory states $\vec{m}$ and having previously received signal vector $s$, the players jointly pick the action profile $\vec{a}$,
and update the memory vector to $\vec{m'}$.

For any two profiles $\sigma,\sigma'\in\BSB{\vec{b}}$, the \emph{distance}  between the strategy profiles is
\[\rdist(\sigma,\sigma')~\eqdef~\max_{\vec{m},s,\vec{m'},\vec{a}}\rdist(p(\sigma,\vec{m},s,\vec{m'}\vec{a}),p(\sigma',\vec{m},s,\vec{m'},\vec{a})).\
\] %
\end{definition}

\subsection{Polynomial witnesses}
\label{ssec:poly-witnesses}
We now show
the existence of suitably small profiles to witness approximate equilibria and zero-sum values.
Our claim that small profiles exist is stated in \cref{lem:small-candidates}. This relies on the following two lemmas,
proofs of which can be found in the appendix.

For the remainder of the section, fix a $k$-player game $G$ with $n$ states and where $C\in\N$ is the largest absolute value of any reward. Also, fix some vector $\vec{b}\in\N^k$ of memory bounds.
Let $N\eqdef n\cdot|S|^k\cdot \prod_j \vec{b}[j]$.

\medskip
The first lemma states that if there are witnesses at all,
then there are also ones where all probabilities are at most double exponentially small.

\begin{lemma}[Lower-bounded probabilities]
    \label{lem:small-probs}
    There is $u_0\in \N$, polynomial in the size of the game and $\epsilon$, so that the following is true
    for all $u\ge u_0$.

    If there exists a profile $\sigma \in \BSB{\vec{b}}$
    that witnesses values $\vec{v} \in\R^k$
    then there is such a profile $\sigma$ 
    where all probabilities
    (for any signal, action and memory mode)
    are at least
    ${1}/{2^{2^u}}$.
\end{lemma}
\begin{proof}
    By \cref{lem:exFO-eq,thm:FOR}.
\end{proof}

Secondly, we show that small perturbations of probabilities
only lead to small changes in (mean-payoff) values for all players.

\begin{restatable}[small perturbations]{lemma}{solan}
\label{lem:solan}

Consider any two profiles $\sigma,\sigma'\in \BSB{\vec{b}}$ and let
$\vec{v},\vec{v'} \in\R^k$ denote the values they witness.

If $\rdist(\sigma,\sigma')<(4N)^{-1}$
then $\vec{v'} \in [\vec{v}-\gamma,\vec{v}+\gamma]$,
where $\gamma=9\cdot N\cdot \rdist(\sigma,\sigma')C$.
\end{restatable}
\begin{proof}
	After fixing the game's strategies to be $\sigma$ it turns into a Markov chain $M$ and similar let $M'$ be the Markov chain obtained by fixing the strategies to be $\sigma'$. The state space of $M$ and $M'$ is the product of the states of the original game, the memory states for each player and the set of signals vectors, i.e, at most $N$ states. Let $\overline{r}$ be the reward function in $M$ and let $\overline{r}'$ be the one in $M'$.
	
        We will apply \cite[Thm.~5]{Sol2003}.%
	This %
        applies generally to two-player, zero-sum concurrent \LAG s and bounds the change 
	in (limit average) values if one changes the reward function additively and the transition function multiplicatively. Here, we apply it to the special case of Markov chains.
	This is because we are interested in what happens if we change the players strategies slightly (which corresponds to changing the transition and cost functions slightly as we will see).
	The following claim states that the reward functions and transition functions differ additively and multiplicatively respectively in the two chains.	
	\begin{restatable}[]{claim}{diff}
	\label{claim:tran-diff}
	Let $(\state,\vec{m},s)\to (\state',\vec{m'},s')$ be an edge of the Markov chains $M$ and $M'$. 
	\begin{itemize}
		\item The difference in rewards in the edge $(\state,\vec{m},s) \to (\state',\vec{m'},s')$ in $M$ and $M'$
		is at most $C\cdot \delta(\sigma,\sigma')$.
		\item 	Let $P$ and $P'$ be the probability to go from $(\state,\vec{m},s)$ to $(\state',\vec{m'},s')$ in $M$ and $M'$ respectively. 
		Then, $\delta(P,P')\leq \delta(\sigma,\sigma')$.
	\end{itemize}
	
\end{restatable}
		
        The proof of \cref{claim:tran-diff} can be found in \cref{app:solan}.
        
        We can now apply \cite[Thm~5]{Sol2003},
        using the claim above and that the maximum absolute reward in $M$ and $M'$ are bounded by $C$, and the number of states in $M$ and $M'$ are
        $N$. %
        We get that the value of $M$ differs from the value of $M'$ by at most $Y\eqdef\frac{4\cdot N\cdot \delta(\sigma,\sigma')}{1-2\cdot N\cdot \delta(\sigma,\sigma')}C+C\cdot \delta(\sigma,\sigma')$.
	Since we assume that  $\delta(\sigma,\sigma')<(4N)^{-1}$, we have that $Y<9\cdot N\cdot \delta(\sigma,\sigma')C$. 
\end{proof}

\begin{lemma}[small candidates]
    \label{lem:small-candidates}
    There exists $u\in \N$,  polynomial in the size of the game and $\eps$, so that the following is true. 

    \begin{enumerate}
        \item 
            If there exists a $\eps$-NE $\sigma$ wrt $\BSB{\vec{b}}$
            then there exists a $2\eps$-NE $\sigma'$ wrt $\BSB{\vec{b}}$
            that is $u$-bit representable .
        \item 
            If there exists a strategy $\sigma$ for player $j$ that $\eps$-achieves ${v}$ wrt $\BSB{\vec{b}}$
            then there exists a strategy $\sigma'$ for player~$j$ that $2\eps$-achieves ${v}$ wrt $\BSB{\vec{b}}$ and is $u$-bit representable.
    \end{enumerate}
\end{lemma}
\begin{proof}
    By \cref{lem:small-probs,lem:solan}.
\end{proof}

\begin{lemma}[Best response]\label{lem:best-response}
Given a game $G$, an $\eps>0$, an $u$-bit represented strategy profile and a player $i$, there is a $u'$ polynomial in the size of the game, $\eps$ and $u$, and a $u'$-bit represented strategy $\sigma'$ for player $i$ that $\eps$-achieves $v$   wrt $\BSB{\vec{b}}$ if there is a strategy for player $i$ that witnesses value $v$   wrt $\BSB{\vec{b}}$
\end{lemma}
\begin{proof}
For a strategy $\sigma_i$ for player $i$, let \[\sigma[\sigma_i]=(\sigma[1],\sigma[2],\dots,\sigma[i-1],\sigma_i,\sigma[i+1],\dots,\sigma[k]).\]
Given a $u$-bit represented strategy profile $\sigma$ we can write a \FOR\ expression for the strategy $\sigma^*$ that maximizes the outcome  for player $i$ of playing $\sigma[\sigma^*]$ in $G$. 

Specifically, that expression could be \begin{align*}
&\exists \sigma^*,y^*\forall \sigma'',y'': \text{Strat}_i(\sigma^*)\wedge \text{Val}_i(\sigma[\sigma^*],y^*)\wedge\\ &((\text{Strat}_i(\sigma'')\wedge \text{Val}_i(\sigma[\sigma'],y''))\Rightarrow y^*\geq y'') ,
\end{align*}
assuming we can express the strategy profile $\sigma$ in a polynomial-sized \FOR\ expression. 
We do that as follows: For each probability $p$ used by $\sigma$, it is expressed in $u$-bit representation, $x\cdot 2^i$, where $x$ and $i$ are a $u$-bit numbers (and $i$ is not positive, because it is a probability). We can simply write $x$ down, but we need to be more careful with $2^i$. We can write down the expressions \[
v_1=1/2\wedge v_2=v_1\cdot v_1\wedge v_3=v_2\cdot v_2\wedge \dots \wedge v_u=v_{u-1}\cdot v_{u-1}
\]
Observe that $v_{\ell}=2^{-2^{\ell}}$. Let $i_j$ be the $j$-th bit set to 1 in the binary expression of $i$ (clearly $j\leq u$) and there are $K$ such bits set for some $K$. We then see that \[
2^i=v_{i_1}\cdot v_{i_2}\cdot \dots\cdot  v_{i_K} .
\]
We can therefore write $p$ as $x\cdot v_{i_1}\cdot v_{i_2}\cdot \dots\cdot  v_{i_K}$ (using another $K$ variables, we can also make these products binary) and we get the full \FOR\ expression by simply doing this for each of the probabilities.

Similar to \cref{lem:small-probs}, this lets us lower bound the probabilities used for strategy $\sigma^*$, by $1/2^{2^{u'}}$ for some polynomial $u'$. We can also apply \cref{lem:solan}, similar to \cref{lem:small-candidates}, from which it follows that there is a $u'$-bit representation of a strategy $\sigma'$ that $\eps$-achieves $v$. 
\end{proof}

\subsection{Computing Approximations for Markov chains}
\label{ssec:approx-mc}
\label{ssec:approximations}

The second ingredient in proving \cref{thm:main-2}
is to approximate mean-payoff values for Markov chains in polynomial time.

Recall that fixing a strategy profile $\sigma\in\BSB{\vec{b}}$ for a game
yields a finite Markov chain where the transition probabilities encode
the players' choices and edges dictate the stepwise rewards for all players.
We aim to approximate the expected mean-payoff
$\ExpOf{\sigma}{\state_0,\signal_0}{\LMP[i]}$
for a given player $i\le k$, initial state $\state_0$ and signal $\signal_0$.

Towards this, consider a Markov chain as discussed in \cref{sec:state-space-reductions},
with states $\{1,\dots,n\}$,
and where $\probm(i,j)$, $\rew(i,j)\in\R$ and $\dur(i,j)\in\R$
denote the probability, reward (for the chosen player), and duration of an edge $i\to j$.
Recall that it suffices to compute the expected \mpr (see \cref{def:mpr})
as initially, the duration of every step is $1$.

We propose an algorithm that exhaustively applies the loop and state-elimination
procedures of \cref{def:state-elim,def:loop-elim},
which is correct by \cref{lem:state-elim,lem:loop-elim}.
However, to deal with the necessarily double-exponentially small values
in input and output Markov chains,
we replace the precise arithmetic operands by
approximate ones that use floating point representations with polynomial many bits.
We show that the error this introduces can be bounded polynomially.

\begin{definition}
A Markov chain is \emph{represented in $u$-bits}
if all probabilities, rewards and durations are in $\+Q(u)$ and
given in $u$-bit floating point representation.

In particular, all distributions $(\probm_{ij})_{j\in\states}$ of a state $i$ is in $\+P(u)$ given in $u$-bit \ANR\ using $u$-bit floating point numbers. %
\end{definition}

We now state the bounds on approximate variants of state and loop eliminations.
Let $\oplus^u,\oslash^{u},\otimes^{u}$ be the $u$-bit finite precision variant of addition, division and multiplication, respectively. These map values $\+Q(u)^2$ to $\+Q(u)$ by rounding down the result of the corresponding arithmetic operation to the nearest value in $\+Q(u)$. We drop the superscript $u$ for readability.

\begin{restatable}[State Elimination]{lemma}{lemappstate}
	\label{lem:approx-state}
        There is a polynomial $\delta_1$ so that the following holds for all $u\in\N$.

	Let $M=(\states,\probm,\rew,\dur)$ 
        and $M'=(\states,\probm',\rew',\dur')$ be Markov chains represented in $u$-bits so that $M'$ results from $M$ by
	 eliminating a transient state $n$ as in \cref{def:loop-elim}
        but using $u$-bit floating point arithmetic.
        That is, for all $i,j\neq n$,
        \begin{align*}
            \probm'(i,j) &= \probm(i,j) \oplus (\probm(j,n)\otimes\probm(n,j))\\
            \rew'(i,j)   &= (\probm(i,j) \otimes \rew(i,j)) \\
                         &\quad\oplus 
                            (
                              (\probm(i,n) \otimes \probm(n,j)) 
                              \otimes
                              (\rew(i,n) \oplus \rew(n,j))
                            )
                            \\
            \dur'(i,j)   &= (\probm(i,j) \otimes \dur(i,j)) \\
                         &\quad\oplus 
                            (
                              (\probm(i,n) \otimes \probm(n,j)) 
                              \otimes
                              (\dur(i,n) \oplus \dur(n,j))
                            )
        \end{align*}
	Then,
	\begin{enumerate}
		\item The smallest (negative) exponent among all floating point numbers in ${M'}$ is at most one smaller than that in $M$.
		\item Then
                    $\abs{\ExpOf{M}{1}{\MPR} - \ExpOf{M'}{1}{\MPR}} \le 
                    \delta_1 2^{-u}$.
	\end{enumerate} 
\end{restatable}

\begin{restatable}[Loop Elimination]{lemma}{lemapploop}
	\label{lem:approx-loop}
        There is a polynomial $\delta_2$ so that the following holds for all $u\in\N$.

	Let $M=(\states,\probm,\rew,\dur)$ 
        and $M'=(\states,\probm',\rew',\dur')$ be Markov chains represented in $u$-bits so that $M'$ results from $M$ by
	 eliminating a loop in state $n$ as in \cref{def:loop-elim}
        but using $u$-bit floating point arithmetic.
        That is,
        $\probm'(n,n)=0$; %
        and for all $i,j\neq n$ let
            $\probm'(i,j) = \probm(i,j)$,
            $\rew'(i,j) = \rew(i,j)$,
            $\dur'(i,j) = \dur(i,j)$
        and 
        \begin{align*}
            \probm'(n,j) &= \probm(n,j) 
                \oslash (\oplus_{k\neq i} \probm(n,k))\\
            \rew'(n,j)   &= ((\rew(n,j) \otimes \probm(n,j)) \oplus ({\probm}_{ii}\otimes \rew_{ii}))\\
                         &\quad \oslash (\oplus_{k\neq i} \probm(n,k))\\
            \dur'(n,j)   &= ((\dur(n,j) \otimes \probm(n,j)) \oplus ({\probm}(n,n)\otimes \dur(n,n)))\\
                         &\quad \oslash (\oplus_{k\neq i} \probm(n,k)).
        \end{align*}
	Then,
	\begin{enumerate}
		\item The smallest (negative) exponent among all floating point numbers in ${M'}$ is at most one smaller than that in $M$.
		\item Then
                    $\abs{\ExpOf{M}{1}{\MPR} - \ExpOf{M'}{1}{\MPR}} \le 
                    \delta_2 2^{-u}$.
	\end{enumerate} 
\end{restatable}

We now show that one can approximate the \mpr 
in polynomial time and sufficiently closely.

\begin{restatable}{lemma}{thmappoxmc}	
\label{thm:approx-mc}
There is a polynomial time algorithm for the following problem.
Given a Markov chain with $n$ states represented in $u$-bits
and error $\eps\in (0,1)$ in binary.

Output a value in $[v,v+\eps)$,
where $v=\ExpOf{M}{1}{\MPR}$
is the expected mean-payoff value from the initial state of $M$.
\end{restatable}
\begin{proof}
W.l.o.g., we can assume that
$u\geq 1000n^2$ and 
$\eps>n(\delta_1+\delta_2)2^{-u}$.
Otherwise, one can increase $u$ to $\log_2((\delta_1+\delta_2)n)\cdot \ell$, where $\ell$ is the number of bits used in the denominator of $\eps$, i.e, $\eps\ge 1/2^{\ell}$.
Since all the constants used are polynomial in the size of the input, so will be $u$.

The algorithm will alternate between 
\begin{enumerate}
    \item exhaustively applying \cref{lem:approx-loop} to remove self-loops in all transient states, and
    \item applying \cref{lem:approx-state} to remove a transient state $n\neq 1$.
\end{enumerate}
This will result in a Markov chain $M'$ where each original CCS is collapsed into a single absorbing state.
Now, if $M'$ contains only a single absorbing state $s$, then 
$\ExpOf{M'}{1}{\MPR} = \rew'(s,s) / \dur'(s,s)$,
so we output
$\rew'(s,s) \oslash \dur'(s,s)$.
Otherwise, $1$ is the only transient state in $M'$
and we output $\oplus_{i\neq 1} \probm'(1,i) \left(\rew'(i,i) \oslash \dur'(i,i)\right)$.

Notice that at most $n-1$ many transient states can be removed, which bounds the number of times we invoke \cref{lem:approx-state}.
Similarly, in between removing states, step (1) results in at most $n$ many applications of \cref{lem:approx-loop}.

Therefore, by \cref{lem:approx-loop}(1) and \cref{lem:approx-state}(1),
the largest exponent in any floating point representation of $M'$
has at most increased by $n(n+1)$ compared to $M$.
Moreover, the total error introduced by these operations is at most
\begin{align*}
    \abs{\ExpOf{M}{1}{\MPR} - \ExpOf{M'}{1}{\MPR}}
    &\le
    n (\delta_1 +\delta_2) 2^{-u}) 
\end{align*}
by \cref{lem:approx-loop}(2) and \cref{lem:approx-state}(2).

By plugging in the value of $u$, we get the error is at most $2^{-\ell}$, therefore at most $\eps$.
\end{proof}

\subsection{Proof of Theorem~\texorpdfstring{\ref{thm:main-2}}{7.3}}
\label{ssec:approx-games}
\label{sec:approx}
We are now ready to prove \cref{thm:main-2}. 

\thmfnp*

We therefore want to make two algorithms in \FNPNP. Because the algorithms are quite similar, we state them both first before proving correctness.

\paragraph{First algorithm (for (1))}
First, for (1), we want to find a $3\eps$-Nash equilibrium  wrt $\BSB{\vec{b}}$ (if an $\eps$-Nash equilibrium  wrt $\BSB{\vec{b}}$ exists). To do so, we use an oracle that solves the following question: Given a game $G$, $\epsilon>0$ and an $u$-bit represented strategy profile $\sigma$, guess an $u'$-bit represented strategy profile $\sigma'$. For each player $i$, construct the induced Markov chain $MC_1^i$ for $\sigma$ and the induced Markov chain $MC_2^i$ for $(\sigma[1],\dots,\sigma[i-1],\sigma'[i],\sigma[i+1],\dots,\sigma[k])$, each with rewards for player $i$.  Apply the algorithm from \cref{thm:approx-mc}, with the value of $\epsilon$ being $\epsilon/3$ and $u$ being $u'$, to each Markov chain, giving approximated values $v_1^i$ and $v_2^i$ respectively. If $v_2^i-v_1^i> \frac{7}{3}\eps$, then return ``yes''. If ``yes'' has not been returned for any player, return ``no''.

The \FNPNP algorithm (for (1)) that uses the oracle is then as follows: Given a game $G$ and an $\epsilon>0$, guess an $u$-bit represented strategy profile $\sigma$ and apply the oracle with $G$, $\epsilon$ and $\sigma$. If the oracle returns ``yes'', return ``no'' and otherwise return $\sigma$.

\paragraph{Second algorithm (for (2))}
Next, for (2), given a value $v$, we want to find a strategy that $3\eps$-achieves $v$ wrt $\BSB{\vec{b}}$ (if a strategy that $\eps$-achieves $v$ wrt $\BSB{\vec{b}}$ exists). To do so, we use an oracle that solves the following question: Given a game $G$, $\epsilon>0$ and an $u$-bit represented strategy $\sigma$ for player $1$, guess an $u'$-bit represented strategy $\sigma'$ for player 2. 
Apply the algorithm from \cref{thm:approx-mc}, with $\epsilon$ being $\epsilon/3$ and $u$ being $u'$, to the Markov chain induced by $(\sigma,\sigma')$, given an approximated value $v'$. If $v-v'>  \frac{7}{3}\eps$, then return ``yes'', otherwise, return ``no''.

The \FNPNP algorithm (for (2)) that uses the oracle is then as follows: Given a game $G$ and an $\epsilon>0$, guess an $u$-bit represented strategy $\sigma$ and apply the oracle with $G$, $\epsilon$ and $\sigma$. If the oracle returns ``yes'', return ``no'' and otherwise return $\sigma$.  

\paragraph{Running time and correctness of the algorithms}

Because the running time of the algorithm from \cref{thm:approx-mc} is polynomial and every $u$-bit represented strategy is also of polynomial size, the oracles are in \NP\ and the algorithms are in \FNPNP.

We thus just need to argue that the algorithms are correct. 

\begin{lemma}
We have the following:
\begin{enumerate}
\item The first algorithm (for (1)) returns either ``no'' or a $3\eps$-Nash equilibrium  wrt $\BSB{\vec{b}}$. Also, if an $\eps$-Nash equilibrium  wrt $\BSB{\vec{b}}$ exists it will not return ``no''.
\item The second algorithm (for (2)) returns either ``no'' or a strategy that $3\eps$-achieve $v$ wrt $\BSB{\vec{b}}$. Also, if a strategy that $\eps$-achieve $v$ wrt $\BSB{\vec{b}}$ exists it will not return ``no''.
\end{enumerate}
\end{lemma}
\begin{proof}

First for (1): We start by arguing that if the algorithm guessed an $u$-bit representable $2\eps$-Nash equilibrium  wrt $\BSB{\vec{b}}$, $\sigma$, which exists if an $\eps$-Nash equilibrium  wrt $\BSB{\vec{b}}$ exists by \cref{lem:small-candidates}, then the oracle will necessarily return ``no'' and therefore such a strategy profile can be returned. Consider first that if we use the real value of $MC_1^i$, then it must be at least the real value of $MC_2^i-2\epsilon$, because $\sigma$ is a $2\eps$-Nash equilibrium. The approximated value of $MC_1^i$ is at most $\epsilon/3$ smaller than the real value. Similar for $MC_2^i$. This means that the approximated value of $MC_1^i$ must be at least the approximated value of $MC_2^i-\frac{7}{3}\epsilon$ and therefore the oracle will necessarily return ``no''.  Because at least one guess would result in an output different from ``no'', the algorithm must return some such value.

Next, consider that the algorithm guessed some $u$-bit representable strategy profile $\sigma$, which was not a $3\eps$-Nash equilibrium  wrt $\BSB{\vec{b}}$ and we will argue that the oracle will be able to return ``yes''.
Let $v$ be the outcome of $\sigma$ for player $i$.
 By definition of $3\eps$-Nash equilibrium  wrt $\BSB{\vec{b}}$, there must be some player $i$ and strategy  wrt $\BSB{\vec{b}}$, $\sigma''$, for player $i$, so that value $v''$ of the Markov chain induced by $(\sigma[1],\dots,\sigma[i-1],\sigma'',\sigma[i+1],\dots,\sigma[k])$ with rewards for player $i$ must be strictly greater than $v+3\eps$. By \cref{lem:best-response}, applied with $\epsilon$ being $\epsilon/3$ we can then find an $u'$-bit representable strategy $\sigma'''$ so that  $(\sigma[1],\dots,\sigma[i-1],\sigma''',\sigma[i+1],\dots,\sigma[k])$ with rewards for player $i$ must have a value strictly greater than $v+\frac{8}{3}\eps$. The oracle can then non-deterministically guess a strategy profile in which $\sigma'[i]=\sigma'''$, the approximated values $v_1^i$ and $v_2^i$ would then be such that $v_2^i-v_1^i> \frac{7}{3}\eps$ (because each is at most $\epsilon/3$ smaller than the real value) and thus, the oracle will return ``yes''.

Next, we will argue that the algorithm for (2) is correct. The argument is similar to the above argument but included for completeness. We start by arguing that if the algorithm guessed a $u$-bit represented strategy $\sigma$ that $2\eps$-achieves $v$ wrt $\BSB{\vec{b}}$ (which exists if a strategy that $\eps$-achieves $v$ wrt $\BSB{\vec{b}}$ exists  by \cref{lem:small-candidates}), then the oracle will return ``no'' and therefore, the algorithm can return such $\sigma$. Consider first that if we use the real value of the induced Markov chain, then it must be at least $v-2\eps$, because the strategy $\sigma$ $2\eps$-achieves $v$. Because we used an approximation of that value, called $v'$, that is at most $\epsilon/3$ smaller than the real value, we get that $v\leq v'+\frac{7}{3}\epsilon$ and thus the oracle will necessarily return ``no''. Because at least one guess would result in an output different from ``no'', the algorithm must return some such value.

Next, consider that the algorithm guessed some $u$-bit representable strategy $\sigma$, which did not 
$3\eps$-achieve $v$ wrt $\BSB{\vec{b}}$ and we will argue that the oracle will be able to return ``yes''.
 By definition of $3\eps$-achieve $v$  wrt $\BSB{\vec{b}}$, there must be some strategy $\sigma''$   wrt $\BSB{\vec{b}}$ for player $2$, so that value $v''$ of the Markov chain induced by $(\sigma,\sigma'')$ with rewards for player $1$ must be strictly smaller than  $v-3\eps$. By \cref{lem:best-response}, applied with $\epsilon$ being $\epsilon/3$ we can then find an $u'$-bit representable strategy $\sigma'''$ so that  $(\sigma,\sigma''')$ with rewards for player $1$ must have a value strictly smaller than $v-\frac{8}{3}\eps$. 
 The oracle can non-deterministically guess that strategy in which $\sigma'=\sigma'''$, the approximated value $v'$ is then such that $v> v'+  \frac{7}{3}\eps$ (because the approximated value is at most $\epsilon/3$ smaller than the real value) and thus, the oracle will return ``yes''.
 \end{proof}

\section{Applications}
\label{sec:applications}
We discuss the implications of our results for some well-known classes of games.

\subsection{Multi-player partial-information parity games}
\label{ssec:imp-parity}

A \emph{parity} objective
(for player $j\le k$) assigns a play
$\rho=~(\state_0, \signal_0,$
$\action_0, \colour_0),(\state_1,\signal_1,\action_1,\colour_1)\dots$
the value $1$ if
\(
    \limsup_{N\to\infty}\colour_i[j]
\), the maximal reward seen infinitely often, is even and $0$ otherwise.

In perfect-information turn-based games, this condition can be translated into a mean-payoff condition, based on the observation that the game effectively ends after a cycle is formed (see e.g.~\cite[Theorem~40]{gamesbook-payoffs}).
However, no such simple reduction can generalise even to concurrent games,
because the most significant reward may occur arbitrarily rarely.
In concurrent games, as opposed to the simpler turn-based ones, $\eps$-optimal strategies might require infinite memory
\cite{HIN2018} and mixed actions with double exponentially small probabilities \cite{HKM2009}.

We now argue that one can adapt our constructions to parity objectives as follows.
First, state and loop elimination (\cref{sec:state-space-reductions})
are simpler for parity since one does not need to keep track of the duration or of the exact reward but can instead just set the new reward to be the maximum of the old ones. 

The only non-trivial change required is the proof of \cref{lem:solan},
which bounds the change in values achieved by strategy profiles under small perturbations.
Specifically, our argument relies on \cite[Theorem 5]{Sol2003}, which only applies to mean-payoff games and not parity games. 
To adapt our proof, we consider the Markov chain induced by any fixed strategy profile where each player uses finite memory.
Now almost surely, some CCS is entered and then every reward contained is seen infinitely often. Therefore a player gets overall parity reward 1 in a CCS within which their largest reward is even (and 0 otherwise).
Therefore, for each CCS, any perturbation of the probabilities in the corresponding states (i.e. state in the original game and memory vector) that have the same support will not change the outcome of that CCS. 
We can therefore define an equivalent mean-payoff Markov chain
in which each CSS is collapsed to an absorbing chain with reward $1$ or $0$ accordingly. The claim of \cref{lem:solan} (parity) then follows from \cref{lem:solan} as stated for mean-payoff.

We conclude that 
\cref{thm:main-1,thm:main-2} hold as stated also for (multi-player, partial-information) games with parity objectives.

\subsection{Multi-player concurrent games} 
We consider different classes of perfect-information games
that are \emph{concurrent} games,
meaning the motion is determined by simultaneously chosen actions
of all players, as opposed to the \emph{turn-based} games where in each round
only the choice of one player matters.

When considering bounded strategies, it is instructive to distinguish  \emph{private} memory (accessible only by the respective player as used throughout this paper) from \emph{public} memory, which is visible to all players.
Essentially, public memory is shared among all players (but can only be updated by its owner) (see \cite[Section 5.1]{HIN2018} for details).

The important difference for us is that
fixing public memory strategies for all but one player results in
an ordinary MDP (a perfect-information 1-player game) for the remaining player.
On the other hand, fixing private memory strategies for all but one player,
results in a \emph{partial-information} MDP, as the player cannot know their opponents' memory modes.

\subsubsection{Computing $\eps$-Nash Equilibria}

Our technique allows us to compute general $\eps$-Nash equilibria, i.e., without restricting the set of admissible strategies, assuming that there exist equilibria with bounded public memory.
Formally, for a $k$-player game and vector $\vec{b}\in \N^k$, let $\PBSB{\vec{b}}$
denote the set of all strategy profiles wherein each player~$i$ uses a
public finite memory strategy with at most $\vec{b}[i]$ modes.
\begin{corollary}\label{cor:concur-Nash}
For any fixed $k\geq 1$, we get the following:
Given a  $k$-player concurrent mean-payoff game, 
bounds $\vec{b}\in N^k$ (in unary), and $\epsilon>0$ (in binary), we can

    \begin{enumerate}
        \item 
            find an $\eps$-Nash equilibrium (and output it) in exponential time, if an $\eps$-Nash equilibrium  wrt $\PBSB{\vec{b}}$ exists
        \item 
            find an $\eps$-Nash equilibrium (and output it) in \FNPNP, if there exists an $\epsilon'$-Nash equilibrium wrt $\PBSB{\vec{b}}$, for some $\epsilon'\in [0,\epsilon/2)$.
    \end{enumerate}
\end{corollary}
\begin{proof}
This follows from \cref{thm:main-1,thm:main-appr}, using the observation that
any $\eps$-Nash equilibrium in $\PBSB{\vec{b}}$ is also an $\eps$-Nash equilibrium without restriction, ie., where unilateral deviations to any strategy are considered, because of the following.
Each player $i$, no matter which bounded and public memory strategy each of the other players picked, player $i$ cannot deviate to a bounded and public memory strategy ensuring more than $\epsilon$ than what player $i$ received before, by definition of $\epsilon$-Nash equilibria with respect to $\PBSB{\vec{b}}$.
Finite MDPs with mean-payoff objectives have memoryless $\eps$-optimal strategies \cite{Put1994}.
Therefore, player $i$ cannot get more than what can be ensured by a memoryless strategy which is a special case of bounded and public memory strategies.
\end{proof}

It is an open question whether $\epsilon$-Nash equilibria always exist for concurrent mean-payoff games.
We point to two special cases: stay-in-a-set games \cite{SS2002}
and
quitting games \cite{SV2001,SV2003}.
For both, it was known that finite (and as we will argue: public) memory $\eps$-Nash equilibria exist,
but so far, no algorithm was known to compute them.

\subsubsection{Stay-In-a-set games}

These are multi-player concurrent games where each player has a separate safety objective,
that is, wants the play to remain in a given subset of the states.
\citeauthor{SS2002} \cite{SS2002} showed that
stay-in-a-set games always have finite memory $\eps$-Nash equilibria, but the players might need to remember who has already lost.
Such strategies are finite memory public strategies (with the number of memories equal to $2^k$, where $k$ is the number of players). 
For a fixed number of players $k$, our technique can be used to
compute $\eps$-Nash equilibria with strategies using public memory and optimise the necessary memory bounds. That is, we can find public-memory equilibria with the least amount of memory.

\begin{corollary}
    For every $k\ge 1$ the following is in \FNPNP.
    Given a $k$-player stay-in-a-set game,
    find an $\epsilon$-Nash equilibria.
\end{corollary}
\begin{proof}
    One first encodes the safety conditions as a mean-payoff condition,
    so that a player receives reward $1$ as long as their set is never left,
    and $0$ forever once it is.
    Using a subset construction, this results in a game with
    states $\states\x 2^k$ (a polynomial blow-up since $k$ is fixed).
    In order to keep a one-to-one correspondence of equilibria it is necessary
    to construct a partial-information game here.
    In particular, in the original game, it is sufficient for players
    to remember who has already lost.
    Our reduction encodes this information into the states, and so it must be hidden from players (by having any state $(s,\vec{m})$ give the signal $s$) because otherwise, memoryless strategies suffice, even if they did not in the original game.
    The claim now follows by \cref{thm:main-2}, bounding the number of memory modes for each player by $b=2^k$, similar to the proof of \cref{cor:concur-Nash}, 
    because, if each player is playing a public and finite-memory strategy, one just needs $\vec{m}$ in memory (requiring $2^k$ memory states) to have all information about the state of the game. 
\end{proof}

\subsubsection{Quitting games}
Quitting games \cite{SV2001,SV2003} are concurrent mean-payoff games
in which there is only one non-absorbing state and the players each have two actions in that state: Continue or Stop.
The play goes to an absorbing state (effectively stops) in the first round when at least one player plays their stop action.

Under mild assumptions on the payoffs\footnote{
They ask that (i) Each player prefers to be the only quitter rather than having the game continue forever, and (ii) if a player quits, he prefers to be the only quitter.},
\citeauthor{SV2001} \cite{SV2001} 
show the existence of ``cyclic'' $\eps$-NE,
in which players' behaviour repeat after some number of stages.
Any such equilibrium can be represented as one consisting of public finite memory strategies. 
Specifically, we need to have a state of memory corresponding to each stage in the ``cyclic'' $\eps$-NE, from stage one up to the first repetition, where we just in each stage move to the next memory state and then loop back at the end.
While they do not provide a bound directly, they prove the existence of such `cyclic'' $\eps$-NE and their proof implies a bound: The proof of \cite[Prop. 2.3]{SV2001} uses a  partitioning of $[-R,R]^k$ (if each reward is bounded absolutely by $R$) into regions, so that any two elements in one region differ by no more than $\eps^2$ (in one-norm). For fixed $k$, this gives on the order of $\?O\left(\frac{R^k}{\epsilon^{2k}}\right)$ many regions. Each region is then mapped to some (perhaps other) region and the cyclic $\eps$-Nash equilibrium is then formed from the sequence obtained by the regions encountered by repeatedly following the mapping. The length of such a sequence can therefore not exceed the number of encountered regions, bounded as above.

Our algorithm can therefore find such an $\eps$-Nash equilibrium, using \cref{cor:concur-Nash},
assuming that both rewards and $\eps$ are encoded in unary.

\begin{corollary}\label{cor:concur-Nash-quit}
For any fixed $k\geq 1$, we get the following:
Given a  $k$-player quitting game (with rewards in unary), satisfying the conditions in \cite{SV2001}, 
and $\epsilon>0$ (in unary), we can find an $\eps$-Nash equilibrium (and output it) in \FNPNP.
\end{corollary}

Similarly, given a game and a cyclic $\eps$-Nash equilibrium, one can compute
a finite memory $\eps$-Nash equilibrium that uses less memory if it exists.

\subsection{$2$-player, zero-sum concurrent mean-payoff games}
\subsubsection{Checking the existence of memoryless ($\eps$-)optimal strategies}
Recent work by 
\citeauthor{BBL2021} \cite{BBL2021,BBL2022,BBL2023} 
studies how to restrict the topology of the game graph
to ensure the existence of memoryless optimal strategies for both players.
They show that it is decidable (in exponential time) whether their
properties hold, and thus $\eps$-optimal strategies exist.
The constructions are based on effective reductions to
the first-order theory of the reals, which also yields
an exponential time algorithm to compute {\mbox{($\eps$-)}}optimal memoryless strategies if they exist.
This has been done for reachability \cite{BBL2021}
B\"uchi and coB\"uchi \cite{BBL2022},
and parity \cite{BBL2023} conditions.

We improve on these results in several ways:
Our constructions are directly applicable 
to not only zero-sum concurrent parity games (see \cref{ssec:imp-parity})
but to more general settings, such as with mean-payoff objectives or where bounded-, finite- and public-memory strategies are sought instead of just memoryless ones.
We can decide the existence of ($\epsilon$-)optimal memoryless strategies for both players
without the proxy of deciding sufficient criteria.
Finally, for $\eps>0$, our technique results in a lower complexity: \FNPNP instead of exponential time. 

\begin{corollary}
    For two-player, zero-sum concurrent mean-payoff games,
    \begin{enumerate}
        \item 
            checking if, for any given game and $\eps\ge 0$,
            there exist $\eps$-optimal memoryless strategies for both players (and output one for each player if so) can be done in exponential time.
        \item 
            checking if, for any given game and $\eps>0$,
            there exist $\eps$-optimal memoryless strategies for both players (and output one for each player if so), if there exists an $\epsilon'$-optimal memoryless strategy for each player, for some $\epsilon'\in [0,\epsilon/2)$ is in \FNPNP.
    \end{enumerate}
\end{corollary}
\begin{proof}
If you have an $\epsilon/2$-optimal strategy for each player in a zero-sum, two-player game, then the strategies form an $\epsilon$-Nash equilibrium and given an $\epsilon$-Nash equilibrium, it must be made up of two strategies that are each $\epsilon$-optimal. Because the strategies are memoryless ($b=1$), they are in particular public memory.
    The two claims thus follow by \cref{cor:concur-Nash} instantiated for $k=2$ and $b=1$. 
\end{proof}

\paragraph{Approximating zero-sum values}
We consider the problem of approximating ordinary (zero-sum) values for two-player concurrent mean-payoff games.
Recall that all strategies getting close to ensuring the value might require infinite memory \cite{HIN2018}.

Different methods (see eg.~\cite{MN1981,KMH2008,HKMT2011,O2021})
have been proposed for this.
The most well-known approach is the value iteration algorithm of
\citeauthor{MN1981} \cite{MN1981}, which has later been show to take double-exponentially many iterations to $\eps$-approximate the values, see~\cite{HKMT2011,HKM2009}. 
The current best complexity upper bound is \PSPACE:
As shown by~\cite{HKMT2011,MN1981} using a double exponentially small discount factor, the corresponding discounted game to a mean-payoff game will have the same value $\pm\epsilon$. As shown by~\cite{EY2008} a discounted game can be solved in \PSPACE\ by reduction\footnote{\citeauthor{EY2008} \cite{EY2008} show how to construct $\exists$\FOR\ formulae for games with fixed discount factors, but one can create double-exponential values using repeated squaring
in the existential fragment.}
to the existential fragment of \FOR.

Our technique yields the current best algorithm for this problem.

\begin{corollary}
    The following is in \FNPNP.
    Given a two-player, zero-sum concurrent mean-payoff game and $\epsilon>0$,
    find the value of the game within an additive error of $\epsilon$.  
\end{corollary}

\begin{proof}
As shown already in \cite{MN1981}, the value of a mean-payoff game is equal to the limit of values of the same game but with discount $d$ for $d$ going to $0$.
\citeauthor{HKMT2011} \cite{HKMT2011} show that for any discount $d$ below some double exponentially small number, the value of the  mean-payoff game differs by less than $\epsilon$ from that of the discounted game with discount factor $d$.
For a formal definition of discounted-payoff games, we refer to \cite{HKMT2011,SV2001}.

There is a well-known reduction from discounted games to (non-discounted) mean-payoff games already present in \cite{S1953}:
Given a $d$-discounted game $G$, replace each action $a$ from $s$ to $t$ with a stochastic action that goes to an absorbing state with reward $r(a)$ with probability $d$ and otherwise (with the remaining probability of $1-d$) proceeds to $t$
(it works the same way if $a$ was part of a stochastic action: If it occurred with probability $p$, then you instead go to the absorbing state with probability $pd$ and to $t$ with probability $p(1-d)$). 
This results in a mean-payoff game $G'$, with transition probabilities in the order of $d$. Importantly, the value and witnessing strategies are the same in both games.
As any concurrent discounted game has optimal memoryless strategies for both players\cite{S1953}, the same is true for $G'$.
Such strategies are in particular public memory strategies. As is well-known, in a two-player, zero-sum game, a Nash equilibrium must be made up of two optimal strategies and an optimal strategy for each player forms an Nash equilibrium.

The claim now follows from \cref{cor:concur-Nash}. %
Note that the corollary is applicable even though $d$ is double-exponentially small because our algorithm allows transition probabilities to be represented in floating point notation.
\end{proof}


\bibliographystyle{ACM-Reference-Format}
\bibliography{bib/journals,bib/conferences,bib/references.clean}

\appendix
\onecolumn  %

\section{Proofs for Section 5}
\label{app:state-space-reductions}
\label{app:MC}

\begin{definition}[Cycle values]
    For a set $C\subseteq\states$ of target states
let $\XHT{C}:\runs\to \N\cup\{\infty\}$ be the random variable
denoting the least non-zero number of steps until state in $C$ is reached. That is,
$\XHT{C}(s_0s_1\ldots) = \inf\{k\ge 1\mid s_k \in C\}$.
Further, let
$\XRS{t},\XDS{t}:\runs\to\R\cup\{\infty\}$
denote the sum of rewards, and durations respectively, accumulated up until state $t$ is first reached. That is,
$\XRS{t} = \sum_{k=1}^{\XHT{t}} \rew(s_{k-1},s_{k})$ and 
$\XDS{t} = \sum_{k=1}^{\XHT{t}} \dur(s_{k-1},s_{k})$.
\end{definition}

\begin{restatable}{lemma}{lemMPReqMCV}
    \label{lem:mc:MPR-eq-MCV}
    If $s$ is recurrent then
    $\ExpOf{}{s}{\MPR} = \ExpOf{}{s}{\XRS{s}} / \ExpOf{}{s}{\XDS{s}}$.
\end{restatable}

\begin{proof}
Let $s$ belong to some CCS $S$ with $n$ states.
 Let $p$ be the smallest non-zero probability associated to any edge between vertices in $S$. Let $R$ (resp. $r$) be the most positive (resp.\ most negative) reward per step and $D$ (resp. $d$) be the largest (resp.\ smallest) duration of a step.
 
 By definition of CCS, there is always a path between any two states in $S$ and thus in particular, a simple one (i.e. with no state occurring twice). This means that there is a simple path consisting of at most $n-1$ edges between the states. Each edge has probability of at least $p$ of occurring, so the path has a probability of at least $p^{n-1}$ to be followed. 
 From each $s'$ consider a shortest (in number of steps) directed path $P_{s'}$ from $s'$ to $s$. We start an {\em attempt} whenever we visit $s$ or just after we have moved along an edge not in $P_{s'}$, where $s'$ was the last state in which we started an attempt. In particular, we start a new attempt every at most $n$ steps. Each attempt reaches $s$ with probability at least $p^{n-1}$. Therefore, in expectation, we need at most $p^{1-n}$ many attempts. Each attempt took at most $n$ steps.
 
This means that almost all plays visits $s$ infinitely many times. Let $s_i$ be the random variable denoting the step in which the play visits $s$ the $i$th time. We view the steps inbetween $s_i$ and $s_{i+1}$ as a {\em trial} and we thus have infinitely many trials.
 
Consider any $\epsilon>0$ .
 We partition plays into two sets, $\Good$ and $\Bad$.  
 A play is in $\Good$, if the average reward of the trials is within $(1\pm \epsilon)$ of $\ExpOf{}{s}{\XRS{s}}$ (unless $\ExpOf{}{s}{\XRS{s}}=0$, in which case, we have that the average reward of the trials is in $[-\epsilon,\epsilon]$ for a play to be in $\Good$) and the average duration of the trials is within $(1\pm \epsilon)$ of $\ExpOf{}{s}{\XDS{s}}$. The remaining plays are in $\Bad$.

By linearity of expectation, we have that \[
\ExpOf{}{s}{\MPR}=\ExpOf{}{s}{\MPR\mid \Bad}\ProbOf{}{s}{\Bad}+\ExpOf{}{s}{\MPR\mid \Good}\ProbOf{}{s}{\Good} .
\]
For any play, $\MPR$ is in $[\min(r/d,r/D),\max(R/d,R/D)]$, because that is the bound on the stepwise fraction. In particular, it is true for $\ExpOf{}{s}{\MPR\mid \Bad}$.  By law of large numbers, $\ProbOf{}{s}{\Bad}<\epsilon$. 

 By definition of $\Good$ we see as follows. If $\ExpOf{}{s}{\XRS{s}}>0$, we have that \[
\ExpOf{}{s}{\MPR\mid \Good}\in\left[\frac{(1-\epsilon)\ExpOf{}{s}{\XRS{s}}}{(1+\epsilon)\ExpOf{}{s}{\XDS{s}}},\frac{(1+\epsilon)\ExpOf{}{s}{\XRS{s}}}{(1-\epsilon)\ExpOf{}{s}{\XDS{s}}}\right]
\]
 Conversely, if $\ExpOf{}{s}{\XRS{s}}<0$, we have that \[
\ExpOf{}{s}{\MPR\mid \Good}\in\left[\frac{(1+\epsilon)\ExpOf{}{s}{\XRS{s}}}{(1-\epsilon)\ExpOf{}{s}{\XDS{s}}},\frac{(1-\epsilon)\ExpOf{}{s}{\XRS{s}}}{(1+\epsilon)\ExpOf{}{s}{\XDS{s}}}\right]
\]
 Finally, if $\ExpOf{}{s}{\XRS{s}}=0$, we have that \[
\ExpOf{}{s}{\MPR\mid \Good}\in\left[-\frac{\epsilon}{(1-\epsilon)\ExpOf{}{s}{\XDS{s}}},\frac{\epsilon}{(1-\epsilon)\ExpOf{}{s}{\XDS{s}}}\right]
\]

In any case, since this is true for any $\epsilon>0$, we have that 
$\ExpOf{}{s}{\MPR}=\frac{\ExpOf{}{s}{\XRS{s}}}{\ExpOf{}{s}{\XDS{s}}}$, as wanted. 
\end{proof}

\begin{definition}\label{def:summary}
    Let $M=(\states,\probm,\rew,\dur)$
    and $M'=(\states',\probm',\rew',\dur')$ be two Markov chains with associated rewards and durations,
    where $\states,\states' \subseteq \N_{\le n}$.
    We call $M'$ a \emph{summary} $M$ if 
    \begin{enumerate}
        \item %
            There is a one-to-one correspondence between CCSs
            $C_0,C_1,\ldots$ in $M$ and
            and $C'_0,C'_1,\ldots$ in $M'$ that agrees on states in $\states\cap\states'$.
            That is, for any pair $C_j$ and $C'_j$ of corresponding CCSs, it holds that
            $C_j\cap\states'=C'_j\cap \states$.
        \item
            The probability of reaching
            any a CCS $C_j$ in $M$ is the same 
            as reaching the corresponding $C'_j$ in $M'$.

            That is,
            $\ProbOf{M}{s}{\Reach{C_j}} = \ProbOf{M'}{s}{\Reach{C'_j}}$
            for any $s\in \states\cap \states'$.
        \item
            The expected cumulative rewards and durations on cycles from and to recurrent states are preserved.
            
            That is, for all $s$ recurrent,
            $
            \ExpOf{M}{s}{\XRS{s}}
            =
            \ExpOf{M'}{s}{\XRS{s}}
            $
            and
            $
            \ExpOf{M}{s}{\XDS{s}}
            =
            \ExpOf{M'}{s}{\XDS{s}}.
            $
    \end{enumerate}
\end{definition}

The following lemma shows that summaries preserve the expected \mpr values.

\begin{lemma}\label{lem:mc:summary-sufficient}
    If $M'$ is a summary of $M$ and $s\in\states\cap\states'$.
    Then
    $\ExpOf{M}{s}{\MPR} = \ExpOf{M'}{s}{\MPR}$.
\end{lemma}
\begin{proof}
Let $B\subseteq\states$ be a maximal set of recurrent states that do not communicate pairwise.
That is, $B$ consists of representative states, one for each CCS of $M$.
Then
\begin{equation}\label{eq:mc:CCS-reach}
    \ExpOf{M}{s}{\MPR} = \sum_{b\in B} \ProbOf{M}{s}{\Reach{{b}}} \cdot \ExpOf{M}{b}{\MPR}.
\end{equation}
This holds because in finite Markov chains the absorption probability is $1$, i.e., almost-surely eventually a state in $B$ will be reached, and that the $\MPR$ objective is prefix-independent.
Now from the assumption that $M'$ is a summary of $M$
we get the following.
Notice that a state $b\in B$ represents a CCS both in $M$ and $M'$ by point 1 of \cref{def:summary}.
\begin{align*}
    \ExpOf{M}{s}{\MPR} 
    &= \sum_{b\in B} \ProbOf{M}{s}{\Reach{{b}}}\cdot \ExpOf{M}{b}{\MPR}  &\text{\cref{eq:mc:CCS-reach}}\\
    &= \sum_{b\in B} \ProbOf{M}{s}{\Reach{{b}}}\cdot \frac{\ExpOf{M}{b}{\XRS{b}}}{\ExpOf{M}{b}{\XRS{b}}} &\text{\cref{lem:mc:MPR-eq-MCV}}\\
    &= \sum_{b\in B} \ProbOf{M'}{s}{\Reach{{b}}}\cdot \frac{\ExpOf{M'}{b}{\XRS{b}}}{\ExpOf{M'}{b}{\XRS{b}}} &\text{\cref{def:summary}}\\
    &= \sum_{b\in B} \ProbOf{M'}{s}{\Reach{{b}}}\cdot \ExpOf{M}{b}{\MPR}
    = \ExpOf{M'}{s}{\MPR} &\text{\cref{eq:mc:CCS-reach}.} &\qedhere
\end{align*}
\end{proof}

It remains to show that the operations of eliminating loops and states create summaries.

\begin{lemma}[Edge Collapse]
	\label{lem:edge-collapse}
	Let $M$ be a Markov chain, $i,j \in\states$ two states with $\probm(i,j)=0$
    and $S\subseteq \states$ be a set of intermediate states $x$ all satisfying 
    $\probm(x,j)=1$ and 
    $\probm(k,x)=0$ for all $k\neq i$.
    Let $M'$ be the Markov chain 
    obtained by simultaneously replacing all length-2-paths from $i$ to $j$ via $S$ by just one direct edge
    with the same expected reward and duration.
    That is, in $M'$ we have that
	\begin{itemize}
        \item $\probm'(i,j)= \sum_{x\in S} \probm(i,x)$
        \item $\rew'(i,j)= \sum_{x\in S} \probm(i,x) (\rew(x,i)+\rew(x,j))$
        \item $\dur'(i,j)= \sum_{x\in S} \probm(i,x) (\dur(x,i)+\dur(x,j))$
	\end{itemize}
Then $M'$ is a summary of $M$.
\end{lemma}
\begin{proof}
    For condition (1) in \cref{def:summary}, notice that any two states $s,t\in\states$ communicate in $M$ iff they do in $M'$.

For point (2), recall that the probabilities $R:\states\to\R$ of reaching a set $C\subseteq \states\setminus S$ of states are the least fixed-point satisfying
\[
    R(s) = 
\begin{cases} 1 &\mbox{if } s\in C\\
    \inf\{ \sum_{k\in\states} \probm(s,k) R(k)\} & \mbox{otherwise} \\
\end{cases}
\]
By definition of $M'$ the sum in the second case can be rewritten as
\begin{align*}
    \sum_{k\in\states} \probm(i,k) R(k)
    &= \sum_{\substack{s\neq i\\k\in\states}} \probm(s,k) R(k)
       + \sum_{\substack{s=i\\ k\in\states\setminus S}} \probm(i,k) R(k) 
       + \sum_{\substack{s=i\\ k\in S}} \probm(i,k) R(k)\\
    &= \sum_{\substack{s\neq i\\k\in\states}} \probm'(s,k) R(k)
       + \sum_{\substack{s=i\\ k\in\states\setminus S}} \probm'(i,k) R(k) 
       + \probm'(i,j) R(j)
     = \sum_{k\in S} \probm'(s,k) R(k)
\end{align*}
We conclude that a least solution of the system of equations is the same for $M$ and $M'$, which implies the claim.

For (3), we show the claim for rewards only; the proof that expected total durations on cycles is analogous.
Pick a recurrent state $s$ and consider the expectation 
Clearly, if state $i$ is not part of the same CCS as $s$ then $\ExpOf{M}{s}{\XRS{s}} = \ExpOf{M'}{s}{\XRS{s}}$.
So suppose now that $s,i$ and $j$ are part of the same CCS.

Note that by assumption on $M$, for every state $k\neq i$ in the same CCS it holds that
$\ExpOf{M}{k}{\XHT{i}} < \ExpOf{M}{k}{\XHT{S}}$
and therefore 
\begin{equation}
\ExpOf{M}{k}{\XRS{i}} = \ExpOf{M'}{k}{\XRS{i}}
\end{equation}

We can split $\ExpOf{M}{i}{\XRS{i}}$ into two parts, according to the disjoint cases or not a state of $S$ is visited in the first step.

The first case is %
\begin{equation}
\label{eq:edge-collapse-2}
   \ExpOf{M}{i}{\XRS{i}\mid \XHT{S}=1}
   = 
   \sum_{k\in S} \probm(i,k)\rew(i,k) + \ExpOf{M}{k}{\XRS{i}}
   =
   \probm'(i,j) +  \ExpOf{M}{j}{\XRS{i}}
   = \ExpOf{M'}{i}{\XRS{i}\mid \XHT{j} = 1}
\end{equation}
The second case is %
\begin{equation}
\label{eq:edge-collapse-3}
   \ExpOf{M}{i}{\XRS{i}\mid \XHT{S}>1}
   = 
   \sum_{k\notin S} \probm(i,k)\rew(i,k) + \ExpOf{M}{k}{\XRS{i}}
   =
   \sum_{k\notin S} \probm'(i,k)\rew(i,k) + \ExpOf{M'}{k}{\XRS{i}}
   = \ExpOf{M'}{i}{\XRS{i}\mid \XHT{j} > 1}
\end{equation}
Using that the two events are disjoint we get 
\begin{equation}
    \label{eq:edge-collapse-1}
\begin{aligned}
    \ExpOf{M}{i}{\XRS{i}}
   &= 
   \ExpOf{M}{i}{\XRS{i}\mid \XHT{S}=1}
    + \ExpOf{M}{i}{\XRS{i}\mid \XHT{S}>1}\\
   &=
   \ExpOf{M'}{i}{\XRS{i}\mid \XHT{j} =1}
    +\ExpOf{M'}{i}{\XRS{i}\mid \XHT{j} > 1}
    =\ExpOf{M'}{i}{\XRS{i}}
\end{aligned}
\end{equation}
For arbitrary $s$ we can write

\begin{equation}
    \label{eq:edge-collapse-4}
\ExpOf{M}{s}{\XRS{s}}
=
\sum_{n\ge 0}\ExpOf{M}{s}{\XRS{i} \mid \XNV{t}{i}=n}
\end{equation}
where $\XNV{t}{i}$ denotes the number of times state $i$ is visited
between time $1$ and $\XHT{t}$.
Clearly the expected reward sum is the same in $M$ and $M'$ on paths that do not visit state $i$.That is,
$\ExpOf{M}{s}{\XRS{t} \mid \XNV{t}{i}=0} = \ExpOf{M'}{s}{\XRS{t} \mid \XNV{t}{i}=0}$.
Now each summand on the RHS of \cref{eq:edge-collapse-4}
can be written as
\begin{align*}
    \ExpOf{M}{s}{\XRS{i} \mid \XNV{t}{i}=n}
&=
    \ExpOf{M}{s}{\XRS{i} \mid \XNV{t}{i}=0}
    +\ExpOf{M}{i}{\XRS{i} \mid \XNV{t}{i}=n}
    +\ExpOf{M}{i}{\XRS{t} \mid \XNV{t}{i}=0}\\
&=
    \ExpOf{M}{s}{\XRS{i} \mid \XNV{t}{i}=0}
    +n\cdot \ExpOf{M}{i}{\XRS{i}}
    +\ExpOf{M}{i}{\XRS{t} \mid \XNV{t}{i}=0}\\
&=
    \ExpOf{M'}{s}{\XRS{i} \mid \XNV{t}{i}=0}
    +n\cdot \ExpOf{M'}{i}{\XRS{i}}
    +\ExpOf{M'}{i}{\XRS{t} \mid \XNV{t}{i}=0}\\
&=
    \ExpOf{M'}{s}{\XRS{t} \mid \XNV{t}{i}=n}
\end{align*}
where the second  equality uses the Markov property and the third one
is due to \cref{eq:edge-collapse-1,eq:edge-collapse-4}.
This concludes the proof that $M'$ is a summary of $M$.
\end{proof}

\begin{figure*}[t]
  \centering
  \begin{tikzpicture}[
            node distance=3cm and 1cm,
            on grid=true
    ]
	\node[state] (4) at (0,0) {$s$};
	\node[state,below left=1.5cm and 1cm of 4] (2) {$n$};
	\node[state, below=of 4] (3) {$t$};
	
        \path (4) edge
            node[lbl,swap,pos=0.2] {$\probm(n,s)$} 
            node[lbl,swap,pos=0.5] {$\rew(n,s)$} 
            node[lbl,swap,pos=0.8] {$\dur(n,s)$}
        (2);

        \path (4) edge 
            node[lbl,pos=0.35] {$\probm(n,t)$}
            node[lbl,pos=0.5] {$\rew(n,t)$}
            node[lbl,pos=0.65] {$\dur(n,t)$}
        (3);
        
        \path (2) edge 
            node[lbl,swap,pos=0.2] {$\probm(n,t)$}
            node[lbl,swap,pos=0.5] {$\rew(n,t)$}
            node[lbl,swap,pos=0.8] {$\dur(n,t)$}
        (3);

	\node[state] (i'') at (6,0) {$s$};

    \node[draw,below left=1.5cm and 1cm of i''] (inj'') {$[s,n,t]$};
    \node[draw,below right=1.5cm and 1cm of i''] (ij'') {$[s,t]$};
	\node[state, below=of i''] (j'') {$t$};
	
        \path (i'') edge[bend right=20]
            node[lbl,swap,pos=0.2] {$\probm(s,n)\probm(n,t)$}
            node[lbl,swap,pos=0.5] {$\rew(s,n)+\rew(n,t)$}
            node[lbl,swap,pos=0.8] {$\dur(s,n)+\dur(n,t)$}
        (inj'');

        \path (inj'') edge[bend right=20]
            node[lbl,swap,pos=0.2] {$1$}
            node[lbl,swap,pos=0.5] {$0$}
            node[lbl,swap,pos=0.8] {$0$}
        (j'');
        \path (i'') edge[bend left=20]
            node[lbl,pos=0.2] {$\probm(n,t)$}
            node[lbl,pos=0.5] {$\rew  (n,t)$}
            node[lbl,pos=0.8] {$\dur  (n,t)$}
        (ij'');
        \path (ij'') edge[bend left=20]
            node[lbl,pos=0.2] {$1$}
            node[lbl,pos=0.5] {$0$}
            node[lbl,pos=0.8] {$0$}
        (j'');

	\node[state] (4') at (12,0) {$s$};
        \node[state,draw=none,below left=of 4'] (2') {};
	\node[state, below=of 4'] (3') {$t$};
	
        \path (4') edge 
            node[lbl,pos=0.35] {$\probm'(n,t)$}
            node[lbl,pos=0.5] {$\rew'  (n,t)$}
            node[lbl,pos=0.65] {$\dur'  (n,t)$}
        (3');
\end{tikzpicture}
  \caption{State elimination:
  On the left is (part of a) Markov chain $M$ before removing state $n$ and on the right is the corresponding part of Markov chain $M'$. In the middle is the intermediate step $M''$ as constructed in the proof of \cref{lem:state-elim}. The probability, and expected duration and reward, of moving from state $i$ to $j$ remains untouched.}
  \label{fig:state_elim_full}
\end{figure*}
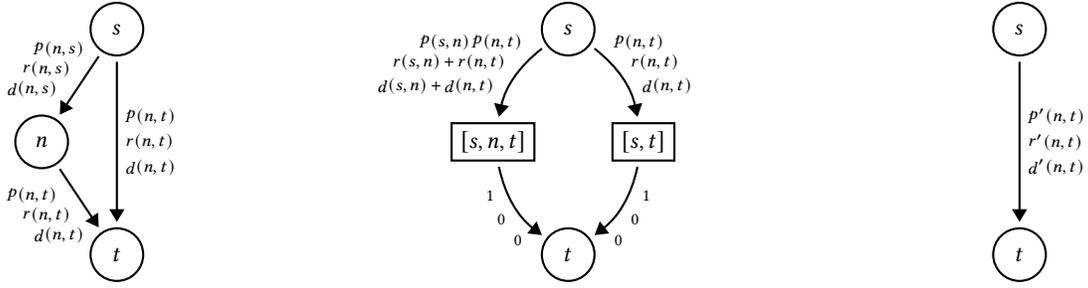
\begin{proof}[Proof of \cref{lem:state-elim}]
    To show that eliminating states preserves the \mpr values we split the transformation from $M$ to $M'$ in two steps via an intermediate Markov chain $M''$ and argue that $M'$ is a summary of $M''$ which in turn is a summary of $M$. 
    The claim then follows from \cref{lem:mc:summary-sufficient}.

    We define $M''$ from $M$ by introducing new intermediate states $[s,n,t]$ and $[s,t]$ between any two states $s,t\in\states$ of $M$, replacing length-two paths from $s$ to $t$ via $n$
    by paths via $[s,n,t]$, and length-one paths from $s$ to $t$ by paths via $[s,t]$.
    See \cref{fig:state_elim_full} for an illustration.
    Formally,
    \begin{align*}
    \probm''(s,n) &= 0;
    \quad
    \probm''(s,[s,n,t]) = \probm(s,n)\probm(n,t); 
    \quad
    \probm''([s,n,t],t) = 1; \text{ and }
    \quad
    \probm''(k,[s,n,t]) = 0 \text{ for all $k\neq s$}.
\end{align*}
The rewards and durations incurred on those steps are
\begin{align*}
    \rew''(s,[s,n,t]) &= \rew(s,t)+\rew(n,t);\\
    \dur''(s,[s,n,t]) &= \dur(s,t)+\dur(n,t); \quad \text{ and }\\
    \rew''([s,n,t],t) &= \dur''([s,n,t],t) = 0.
\end{align*}
Any path from $s$ to $t$ via the new state $[s,t]$ in $M'$
directly corresponds to a length-one path $s\to t$ in $M'$. 
\begin{align*}
\probm''(s,[s,t]) = \probm(s,t);
\text{ and }
\quad
\probm''([s,t],t) = 1
\end{align*}
The rewards and durations incurred on those steps are
\begin{align*}
\rew''(s,[s,t]) = \rew(s,t);
\quad
\dur''(s,[s,t]) = \dur(s,t)
\text{ and }
\quad
\rew''([s,t],t) = \dur''([s,t],t) = 0.
\end{align*}

Notice that for any fixed pair $i,j\neq n$ of states,
$S=\{[i,j], [i,n,j]\}$ satisfies the assumption of \cref{lem:edge-collapse}
and $M'$ is result of collapsing them accordingly. So by multiple applications of \cref{lem:edge-collapse}
we observe that $M'$ is a summary of $M''$.
It remains to show that $M''$ is a summary of $M$.

For (1), just notice that going from $M$ to $M''$ does not change which states in $s,t\in\states$
communicate: they do in $M$ iff they do in $M''$. The newly introduced states $[s,t]$ and $[s,n,t]$
are in a CCS $C$ iff $s,n\in C$.

For (2), first recall that reaching a set $C\subseteq \states$ is an open objective
and thus can be expressed as the countable sum
\begin{equation}\label{eq:se:reachprob}
    \ProbOf{M}{s}{\Reach{C}} = \sum_{\substack{\pi=s_0s_1\ldots s_k\in \states^*\\ s_0=s\land s_k\in C}}\probm(\pi)
\end{equation}
where $\probm(\pi) = \prod_{i=1}^{k} \probm(s_{i-1},s_i)$
is the product of probabilities along $\pi$.
We observe that there is an obvious isomorphism $f:\states^*\to \states(\states'')^*\states$
between finite paths in $M$ that lead from $s$ to $C$, and those
in $M''$ that lead from $s$ to $C$, which preserves probability mass: simply let $f(inj) = i[i,n,j]j$
and $f(ij) = i[ij]j$. Then $\probm(\pi) = \probm''(f(\pi))$.
We can decompose $\ProbOf{M''}{s}{\Reach{C}}$ as in \cref{eq:se:reachprob}
but since neither source nor target consists of states in $\states''\setminus \states$ we have that
\[
    \ProbOf{M}{s}{\Reach{C}}
    = \sum_{\substack{\pi=s_0s_1\ldots s_k\in \states^*\\ s_0=s\land s_k\in C}}\probm(\pi)
    = \sum_{\substack{\pi=s_0s_1\ldots s_k\in \states^*\\ s_0=s\land s_k\in C}}\probm''(f(\pi))
    = \sum_{\substack{\pi''=s_0s_1\ldots s_k\in \states(\states'')^*\states\\ s_0=s\land s_k\in C}}\probm''(\pi'')
    = \ProbOf{M''}{s}{\Reach{C}}.
\]
Point (3) follows analogously to point (2) above. We have that
\[
    \ExpOf{M}{s}{\XRS{s}}
    = \sum_{\substack{\pi=s_0s_1\ldots s_k\in \states^*\\ s_0=s=s_k}}
        \probm(\pi) 
        \cdot \rew(\pi)
    = \sum_{\substack{\pi=s_0s_1\ldots s_k\in \states^*\\ s_0=s=s_k}}
        \probm''(f(\pi)) 
        \cdot \rew''(f(\pi))
    = \ExpOf{M''}{s}{\XRS{s}}
\]
where $\rew(\pi) = \sum_{i=1}^{\len{\pi}} \rew(s_{i-1},s_{i})$
is the total reward along path $\pi$.
\end{proof}
\begin{figure*}[t]
  \centering
  \begin{tikzpicture}[
            node distance=3cm and 0.4cm,
    ]
	\node[state] (2) at (0,0) {$n$};
	\node[state,below left=of 2] (4) {i};
	\node[state, below right=of 2] (3) {$j$};
	
        \draw (2) edge[loop left,looseness=15] node[lbl, align=left] {
            $\probm(n,n)$\\
            $\rew(n,n)$\\
            $\dur(n,n)$
        } (2);
        \path (2) edge
            node[lbl,pos=0.5] {$\probm(n,j)$} 
            node[lbl,pos=0.6] {$\rew(n,j)$} 
            node[lbl,pos=0.7] {$\dur(n,j)$}
        (3);
        
        \path (2) edge 
            node[lbl,swap,pos=0.5] {$\probm(n,i)$}
            node[lbl,swap,pos=0.6] {$\rew(n,i)$}
            node[lbl,swap,pos=0.7] {$\dur(n,i)$}
        (4);
	
	\node[state] (n'') at (6,0) {$n$};
    \node[below left=1.5cm and 1cm of n'' ] (nli'') {$[n,\ell,i]$};
    \node[below right=1.5cm and 1cm of n'' ] (nlj'') {$[n,\ell,j]$};
    \node[state,below left= of n''] (i'') {$i$};
    \node[state,below right= of n''] (j'') {$j$};
	
        \path (n'') edge 
           node[lbl,pos=0.3] {$\probm(n,n)^\ell\probm(n,j)$}
           node[lbl,pos=0.5] {$\ell \cdot\rew(n,n)$}
           node[lbl,pos=0.7] {$\ell\cdot\dur(n,n)$}
        (nlj'');
        \path (nlj'') edge 
            node[lbl,pos=0.1] {$1$}
            node[lbl,pos=0.3] {$\rew(n,j)$}
            node[lbl,pos=0.5] {$\dur(n,j)$}
        (j'');
    
        \path (n'') edge 
           node[lbl,swap,pos=0.3] {$\probm(n,n)^\ell\probm(n,i)$}
           node[lbl,swap,pos=0.5] {$\ell \cdot\rew(n,n)$}
           node[lbl,swap,pos=0.7] {$\ell\cdot\dur(n,n)$}
        (nli'');
        \path (nli'') edge 
            node[lbl,swap,pos=0.1] {$1$}
            node[lbl,swap,pos=0.3] {$\rew(n,i)$}
            node[lbl,swap,pos=0.5] {$\dur(n,i)$}
        (i'');

	\node[state] (30) at (12,0) {$n$};
	\node[state,below left=of 30] (33) {$i$};
	\node[state, below right=of 30] (32) {$j$};
	
        \path (30) edge 
            node[lbl,pos=0.5] {$\probm'(n,j)$}
            node[lbl,pos=0.6] {$\rew'  (n,j)$}
            node[lbl,pos=0.7] {$\dur'  (n,j)$}
        (32);
        
        \path (30) edge 
            node[lbl,swap,pos=0.5] {$\probm'(n,i)$}
            node[lbl,swap,pos=0.6] {$\rew'  (n,i)$}
            node[lbl,swap,pos=0.7] {$\dur'  (n,i)$}
        (33);
\end{tikzpicture}
  \caption{Loop elimination: On the left is (part of the) Markov chain $M$ before eliminating the self-loop in vertex $n$; On the right is the resulting chain $M'$. In the middle is the intermediate chain $M''$ with the countably infinite edges between $n$ and auxiliary state $n'$. Taking the $\ell$th edge represents taking the loop $\ell$ times.}
  \label{fig:loop_elim_full}
\end{figure*}
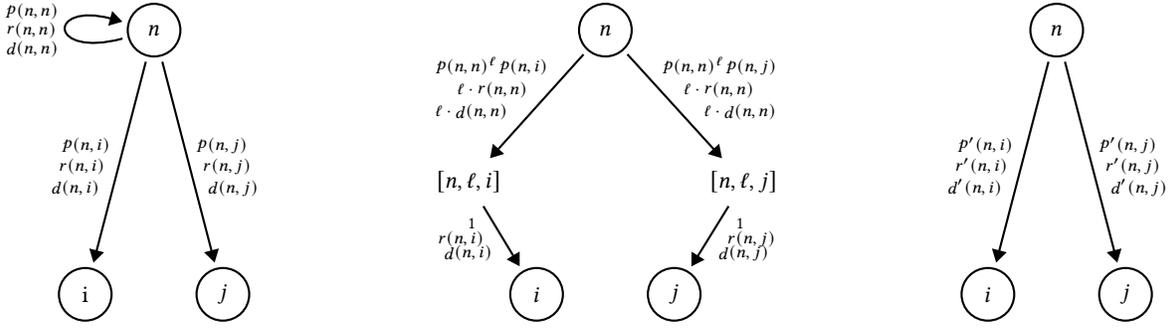

\begin{proof}[Proof of \cref{lem:loop-elim}]
    Consider the Markov chain $M$ and let $n$ be the state whose loop is eliminated.
        We define an intermediate Markov chain $M''$ that replaces the self-loop by
        a countably infinite number of new states, each of which represents a fixed number of iterations of the loop.

        That is $M''$ is obtained from $M$ as follows. 
        For every $\ell\ge 0$ and $j\neq n$ there is a new state $[n,\ell,j]$. %
        We let
        \begin{align*}
            \probm''(n,[n,\ell,j]) &\eqdef \probm(n,n)^{\ell}\probm(n,j)\\
            \rew''(n,[n,\ell,j])&\eqdef\ell\cdot\rew(n,n)\\
            \dur(n,[n,\ell,j]) &\eqdef\ell\cdot\dur(n,n)
        \end{align*}
        and
        $\probm''([n,\ell,j],j)\eqdef 1$ and
        $\rew''([n,\ell,j],j)\eqdef \dur''([n,\ell,j],j)\eqdef 0$.
        
        Every original edge between states $i,j\neq n$ has the same probability, reward and duration as in $M$.
        See also \cref{fig:loop_elim_full} in the middle.

        Intuitively, $M''$ is the (countably infinite) unfolding of the self-loop around state $n$.
        Every path in $M$ that goes from $n$ to $j$ has some fixed number $\ell\ge 0$ of iterations of the loop
        and corresponds to a path $n\to[n,\ell,j]\to j$ in $M''$ which has the same probability, reward and duration. 
        Note in particular that the probability of a path $n\to \cdots \to n\to j$
        that iterates the loop exactly $\ell$ times and then goes to $j$ is
        $\probm(n,n)^{\ell}(1-\probm(n,n)) \cdot \frac{\probm(n,j)}{1-\probm(n,n)}
        = \probm''(n,[n,\ell,j])$.
        One thus readily verifies that $M''$ is a summary of $M$ in the sense of \cref{def:summary}.

        It remains to observe that $M'$ is a summary of $M''$.
        This is because it is the result of removing intermediate states via \cref{lem:edge-collapse}.
        We conclude that $M'$ is a summary of $M$. The claim of the lemma now follows from \cref{lem:mc:summary-sufficient}.
    \end{proof}

\section{Proofs for Section 6}
\label{app:FOR}

\label{app:exFOR}

\lemexFO*
In this section, we show how to find the outcome of an induced Markov chain for a specific player in state 1 in the first order theory over the reals. The problem, in particular, is that depending on which strategy profile we picked, the induced Markov chain could have different (end) strongly connected components and we need the formula to be the same no matter which strategy profile we picked. Still, our solution is to implement loop and state elimination, see Section~\ref{sec:alg-mc}, while taking care that some of these steps might not be possible to do (for example, you cannot do state elimination of a state that is absorbing with incoming edges or if it has already been eliminated,  and you cannot do loop elimination on a state if there is no other successor of the state or if it has been eliminated). 

For each pair of (1)~memory vectors $\vec{m}$ and $\vec{m'}$, (2)~game states $v$ and $v'$ and 
(3)~signal vectors $\vec{s}$ and $\vec{s'}$
and for each iteration $i$ (an iteration consists of either one loop or one state elimination step), 
we have variables corresponding to the (possible) state $(v,m,s)$ and edge $((v,m,s),(v',m',s'))$	 of the induced Markov chain, 
describing the Markov chain at the beginning of iteration $i$. 
For simplicity, we will talk about $(v,m,s)$ as being a state of the Markov chain. 
Therefore, there are at most $N$ states 
and we only need to assign value to the states of the Markov chain initially.

To be explicit about it, we first run a loop and then state elimination on state $x$ for each $x\in\{N,N-1,\dots,1\}$, except we do not do state elimination on state $1$. In other words, for each $1\leq j< 2N$, in iteration $j$ for odd $j$, we do loop elimination on state $N+1-(j+1)/2$ and in iteration $j$ for even $j$, we do state elimination on state $N+1-j/2$.

Specifically, we have variables, each describing the Markov chain at the start of iteration $j$. 
Let $S=(v,\vec{m},\vec{s})$ and $T=(v',\vec{m'},\vec{s'})$. 
The intention of them is as follows:
\begin{enumerate}
\item $x_{S,j}$ should be 1 if state $s$ exists and otherwise 0. 
\item $x_{S,T,j}$ should be 1 if action $(S,T)$ exists and otherwise 0. 
\item $q_{S,T,j}$ should be the probability to use action $(S,T)$ when in $S$. 
\item $n_{S,T,j}$ should be the duration of action $(S,T)$.
\item $r_{S,T,j}$ should be the reward of action $(S,T)$ for the specific player.
\end{enumerate}
We have some expressions (i.e. some boolean combinations of polynomials) giving the value of each of these variables at the start of iteration 1 and depending on the value of the variables at the start of a loop or state elimination iteration, we have expressions giving the value of each of these variables after that iteration. 
The conjunction of all of these expressions give the full formula describing the outcome of the game for the specific player. 

\paragraph{Setting the initial values of the variables}
Note that expressions such as \[\sum_{\vec{a}=(a_1,\dots,a_k)\mid \Delta(v,\vec{a})(*,s',v')>0}\prod_{i=1}^k q_{a_i}^{s_i,m_i,i}b_{m'_i}^{s_i',m_i,i}\] are explicit polynomials given the game and states $(v,m,s)$ and $(v',m',s')$ explicitly because the set of such $a$'s comes directly from the input game ($\Delta(v,a)(*,s',v')$ is the probability to go from $v$ to $v'$ when the players plays $a$, emitting signals $s'$, ignoring the reward) and for a given $a$ the product is explicit. 
The expressions at the beginning, i.e. for iteration 1, are as follows:
\begin{enumerate}
\item The expression for $x_{(v,\vec{m},\vec{s}),1}$ is \[
x_{(v,\vec{m},\vec{s}),1}=1\]
 (one could omit these variables, but including them makes it easier to follow). That is, initially, all states exists.
\item The expression for $x_{(v,\vec{m},\vec{s}),(v',\vec{m'},\vec{s'}),1}$, for $\vec{m}=(m_1,m_2,\dots,m_k)$, $\vec{s}=(s_1,s_2,\dots,s_k)$, $\vec{m'}=(m_1',m_2',\dots,m_k')$ and $\vec{s'}=(s_1',s_2',\dots,s_k')$ is
\begin{align*}
(x_{(v,\vec{m},\vec{s}),(v',\vec{m'},\vec{s'}),1}=1&\wedge \sum_{a=(a_1,\dots,a_k)\mid \Delta(v,a)(*,s',v')>0}\prod_{i=1}^k q_{a_i}^{s_i,m_i}\cdot b_{m'_i}^{s_i',m_i,i}>0) \\
\vee (x_{(v,\vec{m},\vec{m}),(v',\vec{m'},\vec{s'}),1}=0&\wedge \sum_{a=(a_1,\dots,a_k)\mid \Delta(v,a)(*,s',v')>0}\prod_{i=1}^k q_{a_i}^{s_i,m_i,i}\cdot b_{m'_i}^{s_i',m_i,i} =0)
\end{align*}
 In words, the expression sets $x_{(v,m,s),(v',m',s'),1}$ to 1 iff there is an action $a$ that happens with positive probability which has $v'$ while emitting $\vec{s'}$  as an outcome, and the probability to update from $\vec{m}$ to $\vec{m'}$ on signal $\vec{s}$ (i.e. for each player, the probability to update from $m_i$ to $m_i'$ on signal $s_i$) is also positive. 
\item The expression for $q_{(v,\vec{m},\vec{s}),(v',\vec{m'},\vec{s'}),1}$ is \[q_{(v,\vec{m},\vec{s}),(v',\vec{m'},\vec{s'}),1}=\sum_{a=(a_1,\dots,a_k)\mid \Delta(v,a)(*,\vec{s'},v')>0}\prod_{i=1}^k q_{a_i}^{s_i,m_i,i}\cdot b_{m'_i}^{s_i',m_i,i}\cdot \Delta(v,a)(*,s',v')\enspace .\] 
In words, the expression $\prod_{i=1}^k q_{a_i}^{s_i,m_i,i}$ is the probability of using action profile $a$ 
when the signals and memory were $\vec{s}$ and $\vec{m}$ respectively. 
That multiplied by $\prod_{i=1}^{k} b_{m'_i}^{s_i,m_i,a,i}$ is then the probability to update $\vec{m}$ to $\vec{m'}$ as well, on signal $\vec{s}$. 
Finally, that times $\Delta(v,a)(*,s',v')$ is then the probability to go to $v'$, and emitting signal $\vec{s'}$ to the players, when in $v$ with the signal being $s$, and the players used action $a$. 
The full expression is then summing over all possible action profiles $a$.
\item  The expression for $n_{(v,\vec{m},\vec{s}),(v',\vec{m'},\vec{s'}),1}$ is \[n_{(v,\vec{m},\vec{s}),(v',\vec{m'},\vec{s'}),1}=1\enspace .\] I.e. the length of the edge initially is 1.
\item The expression for $r_{(v,\vec{m},\vec{s}),(v',\vec{m'},\vec{s}),1}$ is
 \[r_{(v,\vec{m},\vec{s}),(v',\vec{m'},\vec{s'}),1}q_{(v,\vec{m},\vec{s}),(v',\vec{m'},\vec{s'}),1}=\sum_{r\mid r\text{ is a reward for the selected player}}\sum_{a=(a_1,\dots,a_n)\mid \Delta(v,a)(r,\vec{s'},v')>0}\prod_{i=1}^k q_{a_i}^{s_i,m_i,i}b_{m_i'}^{s_i',m_i,i}\Delta(v,a)(r,\vec{m'},\vec{s'})r\enspace .\] In words, $\prod_{i=1}^k q_{a_i}^{s_i,m_i,i}b_{m_i'}^{s_i',m_i,i}\Delta(v,a)(r,\vec{m'},s')$ expresses the probability that, when in state $v$ and emitting signal $\vec{s}$ and the players have memory $\vec{m}$, we go to $v'$ emitting $s'$ and the players update memory to $\vec{m'}$, while the players used joint action $a$ and the specific player we focus on gets reward $r$. We then sum over all possible $r$'s and $a$'s to get the expected reward.  The reason we have $r_{(v,\vec{m},\vec{s}),(v',\vec{m'},\vec{s'}),1}q_{(v,\vec{m},\vec{s}),
 	(v',\vec{m'},\vec{s'}),1}$ on the left-hand side is that rewards in a Markov chain for an edge are conditional on following that edge (it is perhaps easiest to understand using an example: If the game has only a single action $a$, memory $m$ and signal $s$ for each player in every state and we have that $\Delta(v,a)$, for some state $v$, is the uniform distribution over $(1,s,t)$ and $(1,s,u)$ for some states $t,u$, then the right hand side of the expression for going from $v$ to $t$ is $1/2$ because $\Delta(v,a)(1,s,t)=1/2$ and each other variable in it is 1 (or 0). But clearly, the reward for going from $v$ to $t$ should be 1, because that is the reward for each action in the game!).

\end{enumerate}

For simplicity, in the remainder, we write $S$ for $(v,\vec{m},\vec{s})$ and $T$ for $(v',\vec{m'},\vec{s'})$, because these are the states of the Markov chain.

We will next give the expressions for each variable and each $1<j<2N$, depending on whether $j$ is odd or even (i.e. depending on whether we are doing loop or state elimination).

\paragraph{Loop elimination}

For even $j$ (i.e. we are doing loop elimination from $j-1$ to $j$ on some state $v=N+1-j/2$),  we use the following expressions:
\begin{enumerate}
\item The expression for $x_{S,j}$ is \[x_{S,j}=x_{S,j-1}\enspace .\] In words, each state exists if it did before.
\item 
The expression for $x_{v,v,j}$ is \[
(x_{v,v,j}=0\wedge \sum_{S\mid S\neq v}x_{v,s,j-1}>0)\vee (x_{v,v,j}=1\wedge \sum_{S\mid S\neq v}x_{v,s,j-1}=0)
\enspace .
\]
In words, the loop $(v,v)$ should not exist after this iteration if $v$ has another successor (the sum $\sum_{S\mid S\neq v}x_{v,S,j-1}$ is 0 precisely if for all $S\neq v$, $(v,S)$ does not exist). We can not eliminate a loop if the state does not have another successor though.

For each $S,T$, except $S=T=v$, the expression for $x_{S,T,j}$ is \[
x_{S,T,j}=x_{S,T,j-1} \enspace .
\]
In words, the edge is still there if it was before.
\item 
The expression for $q_{v,S,j}$, for each $S\neq v$ is \begin{align*}
(x_{v,v,j-1}&=1\wedge x_{v,v,j}=0 \wedge q_{v,S,j}(1-q_{v,v,j-1})=q_{v,S,j-1})\vee \\
(x_{v,v,j-1}&=0\wedge x_{v,v,j}=0 \wedge q_{v,S,j}=q_{v,S,j-1})
\end{align*}

In words, if we are doing loop elimination (i.e. the self-loop existed before but not after $x_{v,v,j-1}=1\wedge x_{v,v,j}=0$), then, $q_{v,S,j}=q_{v,S,j-1}/(1-q_{v,v,j-1})\Rightarrow q_{v,S,j}(1-q_{v,v,j-1})=q_{v,S,j-1})$, because $q_{v,v,j-1}\neq 1$ - if the probability of using the self-loop had been 1, we had no other successors of $v$ than $v$ and thus we could not do loop elimination. 
The other part is saying that if there were no self-loop before and after, then the value of $q_{v,S,j}$ is as before. This is fine to do even if the edge $(v,S)$ does not exist, but one could also test that similarly. Note that we do not consider the case where there is a self-loop both before and after even if it can occur: It only occurs in case there are no $S\neq v$ such that $(v,S)$ exists and in that case, we do not need to set $q_{v,S,j}$ to anything in particular.

The expression for $q_{S,T,j}$ for each $S,T$ such that $S\neq v$ is \[
q_{S,T,j}=q_{S,T,j-1} \enspace .
\]
In words, the value of the edge probability is as before. Again, this is fine to do even if the edge $(v,S)$ does not exist, but one could also test that similarly. 

Note that we are not setting $q_{v,v, j}$ at all, since its value does not matter as long as the edge does not exist.
\item 
The expression for $n_{v,S,j}$, for each $S\neq v$ is \begin{align*}
(x_{v,v,j-1}&=1\wedge x_{v,v,j}=0 \wedge n_{v,S,j}(1-q_{v,v,j-1})=q_{v,v,j-1}n_{v,v,j-1}+n_{v,S,j-1}(1-q_{v,v,j-1}))\vee \\
(x_{v,v,j-1}&=0\wedge x_{v,v,j}=0 \wedge n_{v,S,j}=n_{v,S,j-1}) \enspace .
\end{align*}

In words, if we are doing loop elimination (i.e. the self-loop existed before but not after $x_{v,v,j-1}=1\wedge x_{v,v, j}=0$), then, since we took the loop $q_{v,v,j-1}/(1-q_{v,v,j-1})$ many times in expectation, each taking $n_{v,v,j-1}$ many steps, we need to add that in, giving us $ n_{v,S, j}=q_{v,v,j-1}n_{v,v,j-1}/(1-q_{v,v,j-1})+n_{v,S, j}\Rightarrow n_{v,S, j}(1-q_{v,v,j-1})=q_{v,v,j-1}n_{v,v,j-1}+n_{v,S, j}(1-q_{v,v,j-1})$. On the other hand, if we did not have a loop to eliminate, the number of steps for this edge is as before. 
Finally, like in the previous case, if we did not do loop elimination or if this edge does not exist, we can set the value of $n_{v,S, j}$ arbitrarily.

The expression for $n_{S,T,j}$, for each $S\neq v$ is \[
n_{S,T,j}=n_{S,T,j-1} \enspace . 
\]
In words, we do not change the edge at all if it is not going out of $v$.

\item 
The expression for $r_{v,S,j}$, for each $S\neq v$ is \begin{align*}
(x_{v,v,j-1}&=1\wedge x_{v,v,j}=0 \wedge r_{v,S,j}(1-q_{v,v,j-1})=q_{v,v,j-1}r_{v,v,j-1}+r_{v,S,j-1}(1-q_{v,v,j-1}))\vee \\
(x_{v,v,j-1}&=0\wedge x_{v,v,j}=0 \wedge r_{v,S,j}=r_{v,S,j-1}) \enspace .
\end{align*}

In words, this is analogous to how we updated $n_{v,S,j}$

The expression for $r_{S,T,j}$, for each $S\neq v$ is \[
r_{S,T,j}=r_{S,T,j-1} \enspace . 
\]
In words, we do not change the edge at all if it is not going out of $v$.
\end{enumerate}

\paragraph{State elimination}
For odd $j>3$ (i.e. we are doing state elimination from $j-1$ to $j$ on some state $v=N+1-(j-1)/2$),  we use the following expressions:
\begin{enumerate}
\item The expression for $x_{v,j}$ is \begin{align*}
(x_{v,j}&=0\wedge x_{v,j-1}=0)\vee \\
(x_{v,j}&=1\wedge x_{v,j-1}=1\wedge \sum_{S\neq v}x_{S,v,j-1}>0 \wedge \sum_{S\neq v}x_{v,S,j-1}=0)\vee\\ (x_{v,j}&=0\wedge x_{v,j-1}=1\wedge (\sum_{S\neq v}x_{S,v,j-1}=0 \vee \sum_{S\neq v}x_{v,S,j-1}>0))
\end{align*}

In words, we can not eliminate $v$ if it has already been eliminated, or if it is absorbing and have incoming edges.  Note that $\sum_{S\neq v}x_{S,v,j-1}$ is the number of incoming edges to $v$ and $\sum_{S\neq v}x_{v,S,j-1}$ is the number of outgoing from $v$, in both cases ignoring loops.

For each other state $S\neq v$ the expression for  $x_{S,j}$ is \[
x_{S,j}=x_{S,j-1}
\]
In words, it exists if it did before.
\item 
The expression for $x_{v,S,j}$, for each $S\neq v$ is as follows:\[
(x_{v,j}=1\wedge x_{v,S,j}=x_{v,S,j-1})\vee (x_{v,j}=0\wedge x_{v,S,j}=0)
\]
In words, if we did not do state elimination, we do nothing to the edge, otherwise, we remove it.

Similarly, the expression for $x_{S,v,j}$, for each $S\neq v$ is as follows:\[
(x_{v,j}=1\wedge x_{S,v,j}=x_{S,v,j-1})\vee (x_{v,j}=0\wedge x_{S,v,j}=0)
\]

Also, the expression for $x_{v,v,j}$ is as follows:\[
(x_{v,j}=1\wedge x_{v,v,j}=x_{v,v,j-1})\vee (x_{v,j}=0\wedge x_{v,v,j}=0)
\]

Finally, the expression for each other edge $x_{S,T}$ for $S\neq v\neq T$, is as follows: \begin{align*}
(x_{v,j}&=1\wedge x_{S,T,j}=x_{S,T,j-1})\vee \\
(x_{v,j}&=0\wedge x_{S,T,j}=1\wedge (x_{S,T,j-1}=1\vee x_{S,v,j-1}+x_{v,T,j-1}=2))\vee\\
(x_{v,j}&=0\wedge x_{S,T,j}=0\wedge x_{S,T,j-1}=0\wedge x_{S,v,j-1}+x_{v,T,j-1}<2)
\end{align*}
In words, if we did not do state elimination we did nothing, otherwise, there is an edge $(S,T)$ if there were one before or if there were an edge from $S$ to $v$ and one from $v$ to $T$ and otherwise not.
\item 
The expression for $q_{S,T,j}$, for each $S\neq v\neq T$ is \begin{align*}
( x_{v,v,j}&=0 \wedge x_{S,v,j-1}+ x_{v,T,j-1}=2 \wedge x_{S,T,j-1}=1\wedge
q_{S,T,j}=q_{S,T,j-1}+q_{S,v,j-1}*q_{v,T,j-1})\vee\\
(x_{v,v,j}&=0 \wedge x_{S,v,j-1}+ x_{v,T,j-1}<2 \wedge q_{S,T,j}=q_{S,T,j-1})\vee\\
(x_{v,v,j}&=0 \wedge x_{S,v,j-1}+ x_{v,T,j-1}=2 \wedge x_{S,T,j-1}=0\wedge q_{S,T,j}=q_{S,v,j-1}*q_{v,T,j-1})\vee\\
(x_{v,v,j}&=1\wedge q_{S,T,j}=q_{S,T,j-1})
\end{align*}
In words, the first clause is saying that if we do state elimination, we had edges $(S,v),(v,T)$ and $(S,T)$, then the probability to go through $(S,T)$ after the elimination is the probability of going from $S$ to $v$ to $T$ plus the probability to go $(S,T)$ from before we did state elimination. The second clause is saying that if we do state elimination and had $(S,T)$ but not both $(S,v)$ and $(v,T)$ then the probability for $(S,T)$ is as before.
The third is saying that if we do state elimination and had both $(S,v)$ and $(v,T)$, but not $(S,T)$ then the probability for $(S,T)$ after is the probability of going $S$ to $v$ to $T$ before.
The last is saying that if we did not do state elimination, then the probability remains unchanged.

The expression for $q_{v,S,j}$, for each $S$ is \[
(x_{v,v,j}=1\wedge q_{v,S,j}=q_{v,S,j-1})\vee x_{v,v,j}=0
\]
In words, if we did not do state elimination, then the probability remains the same. If we did state elimination we do not care about $q_{v,S,j}$. We could as such equally have used the expression $q_{v,S, j}=q_{v,S,j-1}$, since we do not care about the value of $q_{v,S,j}$ if we did do state elimination, but this makes the construction easier 
to follow.

The expression for $q_{S,v,j}$, for each $S\neq v$ is \[
(x_{v,v,j}=1\wedge q_{S,v,j}=q_{S,v,j-1})\vee x_{v,v,j}=0
\]
In words, if we did not do state elimination, then the probability remains the same. If we did state elimination we do not care about $q_{S,v,j}$. Like above, we could do $q_{S,v,j}=q_{S,v,j-1}$ instead.
\item 
The expression for $n_{S,T,j}$, for each $S\neq v\neq T$ is \begin{align*}
( x_{v,v,j}&=0 \wedge x_{S,v,j-1}+ x_{v,T,j-1}=2 \wedge x_{S,T,j-1}=1 \wedge\\
n_{S,T,i}&(q_{S,T,j-1}+q_{S,v,j-1}q_{v,T,j-1})=n_{S,T,j-1}q_{S,T,j-1}+q_{S,v,j-1}q_{v,T,j-1}(n_{S,v,j-1}+n_{v,T,j-1}))\vee\\
(x_{v,v,j}&=0 \wedge x_{S,v,j-1}+ x_{v,T,j-1}<2 \wedge n_{S,T,j}=n_{S,T,j-1})\vee\\
(x_{v,v,j}&=0 \wedge x_{S,v,j-1}+ x_{v,T,j-1}=2 \wedge x_{S,T,j-1}=0\wedge n_{S,T,j}=n_{S,v,j-1}+n_{v,S,j-1})\vee\\
(x_{v,v,j}&=1\wedge n_{S,T,j}=n_{S,T,j-1})
\end{align*}
In words, the first clause  (which is over the first two lines) is saying that if we do state elimination, we had edges $(S,v),(v,T)$ and $(S,T)$, then the expected length to go through $(S,T)$ after the elimination is \[
\frac{n_{S,T,j-1}q_{S,T,j-1}}{q_{S,T,j-1}+q_{S,v,j-1}q_{v,T,j-1}}+\frac{q_{S,v,j-1}q_{v,T,j-1}(n_{S,v,j-1}+n_{v,T,j-1})}{q_{S,T,j-1}+q_{S,v,j-1}q_{v,S,j-1}}\enspace,\] i.e. if we went through $(S,T)$ after, it corresponded to going from $S$ to $T$ directly before or from $S$ to $v$ to $T$ and those options are not necessarily equally likely, so we need to multiply with the conditional probability of going through those edges. The second clause is saying that if we do state elimination and had $(S,T)$ but not both $(S,v)$ and $(v,T)$ then the length of $(S,T)$ is as before.
The third is saying that if we do state elimination and had both $(S,v)$ and $(v,T)$, but not $(S,T)$ then the length of $(S,T)$ after is the length of the path from $S$ to $v$ to $T$.
The last is saying that if we did not do state elimination, then the length remains unchanged.

The expression for $n_{v,S,j}$, for each $S$ is \[
(x_{v,v,j}=1\wedge n_{v,S,j}=n_{v,S,j-1})\vee x_{v,v,j}=0
\]
In words, if we did not do state elimination, then the length remains the same. If we did state elimination we do not care about $n_{v,S, j}$. We could as such equally have used the expression $n_{v,S, j}=n_{v,S,j-1}$, since we do not care about the value of $n_{v,S, j}$ if we did do state elimination, but like for the probability it seemed harder to follow.

The expression for $n_{S,v,j}$, for each $S\neq v$ is \[
(x_{v,v,j}=1\wedge n_{S,v,j}=n_{S,v,j-1})\vee x_{v,v,j}=0
\]
In words, if we did not do state elimination, then the length remains the same. If we did state elimination we do not care about $n_{S,v, j}$. Like above, we could do $n_{S,v,j}=n_{S,v,j-1}$ instead.

\item 
The expression for $r_{S,T, j}$ is much the same as for $n_{S,T, j}$ but is included for completeness.

The expression for $r_{S,T,j}$, for each $S\neq v\neq T$ is \begin{align*}
( x_{v,v,j}&=0 \wedge x_{S,v,j-1}+ x_{v,T,j-1}=2 \wedge x_{S,T,j-1}=1\wedge\\
r_{S,T,j}&(q_{S,T,j-1}+q_{S,v,j-1}q_{v,S,j-1})=r_{S,T,j-1}q_{S,T,j-1}+q_{S,v,j-1}q_{v,T,j-1}(r_{S,v,j-1}+r_{v,T,j-1}))\vee\\
(x_{v,v,j}&=0 \wedge x_{S,v,j-1}+ x_{v,T,j-1}<2 \wedge r_{S,T,j}=r_{S,T,j-1})\vee\\
(x_{v,v,j}&=0 \wedge x_{S,v,j-1}+ x_{v,T,j-1}=2 \wedge x_{S,T,j-1}=0\wedge r_{S,T,j}=r_{S,v,j-1}+r_{v,T,j-1})\vee\\
(x_{v,v,j}&=1\wedge r_{S,T,j}=r_{S,T,j-1})
\end{align*}
In words, the first clause (which covers the first two lines) is saying that if we do state elimination, we had edges $(S,v),(v,T)$ and $(S,T)$, then the expected reward to go through $(S,T)$ after the elimination is $r_{S,T,j-1}q_{S,T,j-1}/(q_{S,T,j-1}+q_{S,v,j-1}q_{v,T,j-1})+q_{S,v,j-1}q_{v,T,j-1}(r_{S,v,j-1}+r_{v,T,j-1})/(q_{S,T,j-1}+q_{S,v,j-1}q_{v,T,j-1})$, i.e. if we went through $(S,T)$ after, it corresponded to going from $S$ to $T$ directly before or from $S$ to $v$ to $T$ and those options are not necessarily equally likely, so we need to multiply with the conditional probability of going through those edges. The second clause is saying that if we do state elimination and had $(S,T)$ but not both $(S,v)$ and $(v,T)$ then the reward of $(S,T)$ is as before.
The third is saying that if we do state elimination and had both $(S,v)$ and $(v,T)$, but not $(S,T)$ then the reward of $(S,T)$ after is the reward of the path from $S$ to $v$ to $T$.
The last is saying that if we did not do state elimination, then the reward remains unchanged.

The expression for $r_{v,S,j}$, for each $S$ is \[
(x_{v,v,j}=1\wedge r_{v,S,j}=r_{v,S,j-1})\vee x_{v,v,j}=0
\]
In words, if we did not do state elimination, then the length remains the same. If we did state elimination we do not care about $r_{v,S,j}$. We could as such equally have used the expression $r_{v,S, j}=r_{v,S,j-1}$, since we do not care about the value of $r_{v,S,j}$ if we did do state elimination, but like for the probability it seemed harder to follow.

The expression for $r_{S,v,j}$, for each $S\neq v$ is \[
(x_{v,v,j}=1\wedge r_{S,v,j}=r_{S,v,j-1})\vee x_{v,v,j}=0
\]
In words, if we did not do state elimination, then the reward remains the same. If we did state elimination we do not care about $r_{S,v, j}$. Like above, we could do $r_{S,v,j}=r_{S,v,j-1}$ instead.

\end{enumerate}

\paragraph{Finding the value}
Once we did all the above iterations, we are left with a Markov chain where either state 1 is absorbing or it has no self-loop and has edges to absorbing states. All remaining states are absorbing.
We want a variable $v_s$ that, for each state, $s$ is the value of the state (if it exists).
For each state, $s\neq 1$, we have a variable $v_s$ with expression\[
(x_{s,2N}=1\wedge v_sn_{s,s,2N-1}=r_{s,s,2N})\vee x_{s,2N}=0
\]
In words, if the state exists, the mean-payoff value is $v_s=r_{s,s,2N}/n_{s,s,2N}$. Otherwise, if the state does not exist, we do not care.

For state $1$, we have the variable $v_1$ with expression\[
(x_{1,1,2N}=1\wedge v_1n_{1,1,2n}=r_{1,1,2N})\vee 
(x_{1,1,2N}=0\wedge v_1=\sum_{s\neq 1}x_{1,s,2N}q_{1,s,2N}v_{s})
\]
In words, the value of $1$, if it is absorbing, is like the above. If it is not absorbing, it is for each other state $s$ state 1 has an edge to, the probability of going to $s$ from 1 times the value of $s$. Note that $x_{1,s,2N}$ is an indicator variable for whether the edge $(1,s)$ exists.

\paragraph{Proof of lemma}
The correctness of the lemma follows from that each variable, both initially and based on the previous variables, is set correctly. There are $2N-1$ iterations. In each iteration, we use $5$ variables, 4 for edges and 1 for states. There can be at most $N^2$ edges and $N$ states. Because we have variables for before and after each iteration (the before matches the ones for after the previous), we get a total of $8N^3+2N$ variables for this. Additionally, we use $N$ variables for the values, getting us up to $8N^3+3N$. Each variable requires between $1$ and $13$ polynomials for below $100N^3+40N$ polynomials in total ($8\cdot 13=104$, but many of the expressions are much smaller than the worst case). The maximum degree is $|A|\cdot b$ for each of these polynomials (for setting the initial variables). The worst case for the numbers in the coefficients comes from the input game of $\tau$.

\section{Proofs for Section 7}
 \label{app:approx}
 \subsection{A proof of Claim~\texorpdfstring{\ref{claim:tran-diff}}{7.4}}
\label{app:solan}

\diff*
\begin{proof} 
	We first prove that the rewards differ additively.
	Without loss of generality, we assume that for every pair of game states $\state,\state'\in \states$, actions vector $\vec{a}\in \actions^k$, and signal vector $s$, there exists a unique reward vector $\vec{r}$ such
	that $\trans(\state,\vec{a})(\vec{r},s,\state')>0$. If the same action profile produces different rewards on the same action, then we can replace all such transitions in the game by a single transition which produces the expected reward from all such possible transitions.
	Let 
	 \[X\eqdef 
	 \sum_{	 	a\in \actions^k \mid \Delta(\state,a)(\vec{r},\vec{s'},v')>0} 
 	\Delta(\state,a)(\vec{r},\vec{s'},\state')\vec{r}[i]\cdot
 	\left(p(\sigma,\state,\vec{m},s,\vec{m'},\vec{a})- p(\sigma',\state,\vec{m},s,\vec{m'},\vec{a})\right) \]
be the difference in rewards in $M$ and $M'$ for player $i$.
For convenience, we write $p(a)$ instead of $p(\sigma,\vec{m},s,\vec{m'},\vec{a})$ and $p'(a)$ for $p(\sigma',\vec{m},s,\vec{m'},\vec{a})$.%

We have that $\vec{r}[i]\in [-C, C]$ and we will next bound $p(a)-p'(a)$. 
We have that for all action profiles $a$, 
$p(a)-p'(a)= p(a)(1-p'(a)/p(a))$. There are two cases, either $p'(a)>p(a)$ or $p(a)>p'(a)$. 
If $p'(a)>p(a)$, then $1-p'(a)/p(a)=-\delta(p(a),p'(a))$. Otherwise, if $p(a)>p'(a)$ then $\delta(p(a),p'(a))=p(a)/p'(a)-1$, implying that $p'(a)/p(a)=1/(\delta(p(a),p'(a))+1)$ and thus that $1-p'(a)/p(a)=\delta(p(a),p'(a))/(\delta(p(a),p'(a))+1)$. 
We have that $p(a)/p'(a)+1)\leq \delta(p(a),p'(a))$ because $\delta(p(a),p'(a))$ is non-negative and we, therefore, divide with something bigger than 1. Hence, $p(a)-p'(a)\in [-p(a)\delta(p(a),p'(a)),p(a)\delta(p(a),p'(a))]$. Clearly, $\delta(p(a),p'(a))\leq \delta(\sigma,\sigma')$, because it is a maximum over a set containing $\delta(p(a),p'(a))$.
Thus, $\abs{X}\leq \sum_{a=(a_1,\dots,a_k)} C\cdot \trans(\state,a)(\vec{r},s,\state') \cdot p(a) \delta(\sigma,\sigma')=C\cdot \delta(\sigma,\sigma')\sum_{a=(a_1,\dots,a_k)}\trans(\state,a)(\vec{r},s,\state')p(a) $, which is $\leq C\cdot \delta(\sigma,\sigma')$ as the sum of probabilities is $\leq 1$.

We now prove that the transition probabilities differ multiplicatively.
We have that 
$P=\sum_{a=(a_1,\dots,a_k)}p(\sigma,\vec{m},s,\vec{m'},\vec{a})\cdot \trans(\state,\vec{a})(r,\vec{s'},\state')$ and similarly $P'=\sum_{a=(a_1,\dots,a_k)}p(\sigma',\vec{m},s,\vec{m'},\vec{a})\cdot \trans(\state,\vec{a})(r,\vec{s'},\state')$. 
We will look at the sums term by term. 
We have three cases, either $P=P'$, $P>P'$ or $P<P'$. In the first case, we have that $\delta(P,P')=0\leq \delta(\sigma,\sigma')$.
The two remaining cases we show just one as the other is analogous.
Consider that $P>P'$ and thus $\delta(P,P')=P/P'-1$. 
We have that  \[
p(a)/p'(a)-1\leq \delta(\sigma,\sigma')\Rightarrow p(a)\leq (\delta(\sigma,\sigma')+1)p'(a)\]
and therefore that
\begin{align*}
\delta(P,P')+1=P/P'
&=\frac{
\sum_{a=(a_1,\dots,a_k)}p(a)\cdot \trans(\state,a)(\vec{r},s,\state')
}
{
\sum_{a=(a_1,\dots,a_k)}p'(a)\cdot \trans(\state,a)(\vec{r},s,\state')
}\\
&\leq 
\frac{
\sum_{a=(a_1,\dots,a_k)}p'(a)(\delta(\sigma,\sigma')+1)\cdot \trans(\state,a)(\vec{r},s,\state')
}
{
\sum_{a=(a_1,\dots,a_k)}p'(a)\cdot \trans(\state,a)(\vec{r},s,\state')
}\\
&=\delta(\sigma,\sigma')+1\qedhere
\end{align*}

\end{proof}

\subsection{Approximating Markov Chains}
 \label{app:approx}
 We first recall some properties of  $u$-bit floating point numbers $\+Q(u)$ and the finite-precision variants of the  addition, multiplication and division operations $\oplus^u,\oslash^{u},\otimes^{u}$.

\begin{proposition}
	\label{prop:uclose}\
	\begin{enumerate}
		
		\item\label{prop:uclose-exists-close} %
		For every $x\in\nnreals$ and $u\in\N$ there exists $x'\in\+Q(u)$ that is $(u,1)$-close to $x$.
		
		\item\label{prop:uclose-trans} %
		If $x$ is $(u,i)$-close to $y$ and $y$ is $(u,j)$-close to $z$, then 
		$x$ is $(u,i+j)$-close to $z$.
		
		\item\label{prop:uclose-ops} %
		Let $x$, $\tilde{x}$, $y$, $\tilde{y}$ be non-negative numbers 
		such that $x$ is $(u,i)$-close to $\tilde{x}$ and $y$ is
		$(u,j)$-close to $\tilde{y}$. 
		Then, $x+y$ is $(u,\text{max}(i,j)+1)$-close to $\tilde{x}+\tilde{y}$,
		$xy$ is $(u,i+j)$-close to $\tilde{x}\tilde{y}$, and 
		$x/y$ is $(u,i+j)$-close to $\tilde{x}/\tilde{y}$.

		\item\label{prop:uclose-imprec-ops} %
		If $x'\in \+Q(u)$  is $(u,j)$-close to $x$ and $y'\in \+Q(u)$ is $(u,j)$-close to $y$, then
		$x'\oplus^uy'$ is $(u,\max(i,j)+1)$-close to $x+y$;
		$x'\oslash^uy'$ is $(u,ij+1)$-close to $x/y$;
		and
		$x'\otimes^uy'$ is $(u,i+j+1)$-close to $xy$;
		
	\end{enumerate}
\end{proposition}

We also recall some properties of probability distributions represented by \ANR.
\begin{lemma} 
	\label{lem:prob-reps}\
	
	\begin{enumerate}
		\item\label{lem:prob-reps-exists}
		For any probability distribution $q=(q_1,q_2,\ldots,q_m)$
		there exist
		$(p_1,p_2,\ldots,p_m)\in\+P(u)$ so that for all $i$,
		$q_i$ and $p_i$ are $(u,2m+2)$-close.
		
		\item\label{lem:prob-reps-comp}
		Suppose that 
		$a_1,a_2,\ldots,a_m\in\+D(u)$ and 
		for all $i=1\ldots m$ let
		$p'_i=a_i\oslash^u\bigoplus^u_{j=1\ldots m} a_j$
		and 
		$p_i=p'_i / \sum_{j=1\ldots m} p'_j$.
		Then 
		$(p'_1,p'_2,\ldots,p'_m)$ is an \ANR\ of
		$(p_1,p_2,\ldots,p_m)\in\+P(u)$.
		
	\end{enumerate}
\end{lemma}

We recall the following Lemma of Solan~\cite{Sol2003}.
This is very similar to our Lemma~\ref{lem:solan} in that it quantifies
the change in value as a result of perturbations.
Instead of perturbing strategies in MDPs, as we do in \cref{lem:solan},
this considers perturbations of transition probabilities \emph{and reward functions} in Markov chains.

\begin{lemma}[\cite{Sol2003}, Theorem~4]
    \label{lem:solan-original}
	Given a Markov chain $M$ and $M'$ with the same states, let $\rdist(\probm_{ij},\probm_{ij})\leq \eps$,   
	$\rdist(\rew_{ij},\rew'_{ij})\leq C$, and $\rdist(\dur_{ij},\dur'_{ij})\leq C$. Then the mean payoff value of $M$ and $M'$ differ by at most $4\eps NC$. 
\end{lemma}  
We can now bound the error introduced by the imprecise loop elimination and state eliminations procedures.

\lemapploop*
\begin{proof}
	Let $\hat{M}$ be the Markov chain obtain by 
	the precise loop elimination algorithm.
	By Lemma~\ref{lem:loop-elim}, $M$ has the same 
	mean-payoff value as $\hat{M}$. 
	Therefore, we bound the difference in mean-payoff value of
	$M'$ and $\hat{M}$.
	Let $N$ be the number of states in the Markov chain.
	
	Since the probabilities in $M$ and $M'$ are given in \ANR, 
	we write $\probm(i,j)$ and $\probm'(i,j)$ for the numbers 
	giving the \ANR\ of the distributions 
	$\tilde{\probm}(i,j)$ and $\tilde{\probm}'(i,j)$ respectively. 
	We write $\hat{\probm}(i,j)$ to represent the probabilities in $\hat{M}$.

	\begin{enumerate}
		\item Each $\probm(i,j)$ is $(u,N)$-close to the actual probability distribution of $M$, which is $\tilde{\probm}(i,j)$ by \cref{lem:prob-reps}.
		\item Each $\oplus_{k\neq i} \probm(n,k)$ is $(u,2N-1)$-close to $\sum_{k\neq i}  \tilde{\probm} (n,k)$ by $N-1$ applications of \cref{prop:uclose-imprec-ops} of \cref{prop:uclose} and item (1) above.
		\item Each $\probm'(n,j)$ is $(u,(2N-1)N)$-close to $\hat{\probm}(n,j)$ by \cref{prop:uclose-imprec-ops} of \cref{prop:uclose} and items (1) and (2) above.
		\item Each $\tilde{\probm}'(n,j)$ is $(u,N)$-close to $\probm'(n,j)$ by definition of \ANR\ and \cref{lem:prob-reps}.
		\item Each $\tilde{\probm}'(n,j)$ is $(u,2(N-1)N+N)$-close to $\hat{probm}(n,j)$ by \cref{prop:uclose-trans} of \cref{prop:uclose}
		and items (3) and (4), and therefore at least $(u,2N^2)$-close.
		\item Each $(\rew(n,j)\otimes \probm(n,j)\oplus(\rew(n,n))\otimes\probm(n,n))$ is 
		$(u,N+1)$-close to 
		$\rew(n,j) \tilde{\probm(n,j)} + \rew(n,n)\tilde{\probm}(n,n)$ by item (1) above and \cref{prop:uclose-imprec-ops} of \cref{prop:uclose}.
		\item  Each $\rew'(i,j)$ is $(u,2(N-1)(N+1))$-close to $\hat{\rew}(i,j)=(\rew(n,j) \tilde{\probm(n,j)} + \rew(n,n)\tilde{\probm}(n,n))/\sum_{k\neq i} \tilde{\probm}(n,k)$
		by \cref{prop:uclose-imprec-ops} of \cref{prop:uclose} applied to items (6) and (2), so at most $(u,2N^2+N)$-close.
		\item  Each $\dur'(i,j)$ is $(u,2N^2 +N)$-close to $\hat{\dur}(i,j)$ by the same argument as item (7).
	\end{enumerate}
	
	By Lemma \ref{lem:solan-original}, 
	we get that the mean-payoff value of $M'$ differ from the value in $\hat{M}$ by at most $4\eps NC$, where the $\eps$ is the maximum difference in the probabilities of $\hat{M}$ and $M'$,
	and $C$ is the maximum difference of rewards in $\hat{M}$ and $M'$. In our case, we have 
	$\eps= \big(\frac{1}{1-2^{-u-1}}\big)^{2(N)^2}-1$. 
	Since, $u\geq 1000N^2$, we have 
	$\big(\frac{1}{1-2^{-u-1}}\big)^{2(N)^2}\le 5(N^2)2^{-u}$, which gives $\eps\le 5(N^2)2^{-u}$.
	Similarly, we get $C\le 5(N^2+N)2^{-u}$, which together gives 
	$4N\eps C\le 100(N^4+N^3)2^{-u}$. Choosing $\delta_2=100(N^4+N^3)$ completes the proof of item (2).
	
	The exponent for the probabilities are always negative as the values are $\leq 1$. It is easy to see that in our updates $2\tilde{\probm}_{ij}\geq \tilde{\probm'}_{ij}\geq \tilde{\probm}_{ij}$. This implies that the negative exponents in the probabilities decreases by at most $1$. 
	
	The rewards and durations can be assumed to be $\geq 1$ and therefore the exponent can be assumed to be positive. Note that the rewards can be made $>1$ by adding the smallest reward to each step and then subtracting it from the obtained value in the end. In a single step the rewards are less than twice the maximum reward in the input Markov chain. Therefore, the maximum exponent is  at most 
	$1$ more than that in the input. A similar argument works for duration as well assuming that the durations in the input are all $\geq 1$.
\end{proof}

\lemappstate*

\begin{proof}
	
	Let $\hat{M}$ be the Markov chain obtain by 
	the precise state elimination algorithm.
	By Lemma~\ref{lem:state-elim}, $M$ has the same 
	mean-payoff value as $\hat{M}$. 
	Therefore, we bound the difference in mean-payoff value of
	$M'$ and $\hat{M}$.
	We again estimate the relative distance between the precise and imprecise variables using the same notations as above.
	Let $N$ be the number of states of the Markov chain.
	
	\begin{enumerate}
	\item Each $\probm(i,j)$ is $(u,N)$-close to the actual probability distribution of $M$, which is $\tilde{\probm}(i,j)$ by \cref{lem:prob-reps}.
	\item Each $\probm'(i,j)$ is $(u,2N+2)$-close to $\hat{\probm}(i,j)$ by item (1) above and applying \cref{prop:uclose-imprec-ops} of \cref{prop:uclose}.
	\item Each $\rew'(i,j)$ is $(u,2N+3)$-close to $\hat{\rew}(i,j)$ by applying \cref{prop:uclose-imprec-ops} of \cref{prop:uclose} to item (1). The dominating factor is the distance between $\probm(i,n)\otimes \probm(n,j)\otimes (\rew(i,n)\oplus \rew(n,j))$ and
	$\tilde{\probm(i,n)}\tilde{\probm(n,j)}(\rew(i,n)\oplus \rew(n,j))$, which gives the $2N+3$ term.
	\item Similar to item (3), we have $\dur'(i,j)$ is $(u,2N+3)$-close to $\hat{\dur}(i,j)$.
\end{enumerate}
	
	Instead of using the bound $2N+3$ obtained above, we instead use $3N$ for the sake of brevity, as $(u,2N+3)$-closeness implies $(u,3N)$-close assuming $N\ge 3$. 
	By Lemma \ref{lem:solan-original}, we get that the mean-payoff value of $M'$ differ from the value in $\hat{M}$ by at most $4N \eps C$, where the $\eps$, $\eps= \big(\frac{1}{1-2^{-u-1}}\big)^{3N}-1$. Since $u\ge 1000N^2$, we have $\big(\frac{1}{1-2^{-u-1}}\big)^{3N} \leq 1 + 7N 2^{-u}$ which gives $\eps\le 1+ 7N 2^{-u}$. We get the same bound on $C$ and therefore, $4N\eps C\le 200N 2^{-u}$. 
	Choosing $\delta_1=200N$ completes the proof of item (2).
	
	The increase in positive and decrease in negative exponents can be shown to be at most $1$ by the same argument as the loop-elimination case.
\end{proof}

\end{document}